\documentclass[twocolumn,aps,pra,superscriptaddress,nofootinbib]{revtex4-2}
\usepackage{dcolumn}
\usepackage{bm,bbm}
\usepackage[T1]{fontenc}
\usepackage[utf8]{inputenc}
\usepackage{tikz}
\usepackage{amssymb,amsmath,amsthm,amsbsy,mathptmx}
\usepackage{amsfonts}
\usepackage{mathtools}
\usepackage{physics}
\DeclareMathAlphabet{\mathcal}{OMS}{cmsy}{m}{n}
\usepackage{tabularx}
\usepackage{xcolor}
\usepackage{oubraces}
\usepackage{dsfont}
\usepackage[title, titletoc]{appendix}
\usepackage{graphicx}
\usepackage{rotating}
\usepackage[most]{tcolorbox}
\usepackage{enumerate}
\usepackage{subfloat}
\usepackage[normalem]{ulem}
\usepackage[colorlinks=true,linkcolor=blue,citecolor=blue,urlcolor=blue]{hyperref}%
\usepackage{cleveref}
\usepackage[normalem]{ulem}
\Crefname{subfigures}{figure}{figures}
\Crefname{subfigures}{Figure}{Figures}

\def\bra#1{\mathinner{\langle{#1}|}}
\def\ket#1{\mathinner{|{#1}\rangle}}

\newtheorem{theorem}{Theorem}

\newtheorem{corollary}[theorem]{Corollary}

\newtheorem{proposition}{Proposition}

\begin{document}

\title{Golden states in resource theory of superposition}

\author{H\"{u}seyin Talha \c{S}enya\c{s}a}
\email{senyasa@itu.edu.tr}
\affiliation{Department of Physics, \.{I}stanbul Technical University, 34469 Maslak, \.{I}stanbul, Turkey}

\author{G\"{o}khan Torun}
\email{gtorun@ku.edu.tr}
\affiliation{Department of Physics, Ko{\c{c}} University, 34450 Sariyer, \.{I}stanbul, Turkey}


\begin{abstract}
One central study that constitutes a major branch of quantum resource theory is the hierarchy of states. This provides a broad understanding of resourcefulness in certain tasks in terms of efficiency. Here, we investigate the maximal superposition states, i.e., golden states, of the resource theory of superposition. Golden states in the resource theory of coherence are very well established; however, it is a very challenging task for superposition due to the nonorthogonality of the basis states. We show that there are sets of inner product settings that admit a golden state in high-dimensional systems. We bridge the gap between the resource theory of superposition and coherence in the context of golden states by establishing a continuous relation by means of a Gram matrix. In addition, immediate corollaries of our framework provide a representation of maximal states which reduces to the maximal state of the coherence in the orthonormal limit of pure basis states.
\end{abstract}

\maketitle


\section{Introduction}

The superposition principle \cite{Dirac-Superposition} is a characteristic feature of quantum mechanics and of particular interest for quantum information theory. Clearly, as we have witnessed with seminal and far-reaching applications, this and other types of quantum properties such as entanglement \cite{Horodecki-QE}, coherence \cite{Baumgratz-Coherence}, imaginarity \cite{Hickey_2018}, and Bell nonlocality \cite{2014Bellnonlocality} could produce some remarkable differences between the two regimes --- quantum and classical --- in terms of how effectively the information processing tasks are achieved. Over the past two decades, this reasoning has provided an impetus to investigate the role of quantum effects as a resource in certain tasks, focusing on whether these lead to any operational advantages \cite{Winter120404,Regula-QRTs,Oszmaniec2019operational,Liu-OneShotRT,Takagi_OpAdvQR,RT-of-QC,Li-QuantifyingRCQC,Wu2021-Imaginarity}. In particular, their performance as a resource reflects the order of states in terms of resourcefulness, where such orderings can be characterized, for instance, by consistent resource measures \cite{Kuroiwa2020generalquantum} or by state conversion protocols \cite{Zhou2020SC}.

Questioning the existence of a golden state is one of the most notable problems of quantum resource theories (QRTs) in that it provides a practical perspective for the quantification of quantum resources. In simple terms, a \(d\)-dimensional maximal state is identified as a state that can be transformed into any other state in dimension \(d\) by means of the free operations. In this sense, a maximal state is a state at the top of the hierarchy of resourcefulness, having the greatest resource value and regarded as a (golden) unit of the resource. Such golden states are well defined in the resource theory of (bipartite) entanglement, \(\ket{\Psi_{d}}=({1}/{\sqrt{d}}) \sum_{i=1}^{d}\ket{ii}\) \cite{Horodecki-QE} (see \cite{Tejada2019MRE} for multipartite systems); coherence, \(\ket{\Psi_{d}}=({1}/{\sqrt{d}}) \sum_{i=1}^{d}\ket{i}\) \cite{Baumgratz-Coherence,Peng2016MCS}; and imaginarity, \(\ket{\hat{\mp}}=(\ket{0} \mp i\ket{1})/\sqrt{2}\) \cite{Hickey_2018}. Recently, in investigating the general properties of QRTs, the authors of Ref.~\cite{Kuroiwa2020generalquantum} emphasized that the compactness of the set of free operations is an essential prerequisite that indicates the existence of maximally resourceful states. Although they further proved that a type of maximally resourceful state exists in the general QRTs satisfying axioms (see \cite{Kuroiwa2020generalquantum}, Sec. 2.2, for details) and the compactness of the set of states, they mentioned that the existence of a maximal state in QRTs is not obvious in general. A paucity of information also exists  for superposition (of nonorthogonal basis states); our paper describes an approach to adress this.

While coherence and superposition are often considered the same, the latter is a generalization of the former in the context of resource theory. A rigorous resource theory framework for the quantification of superposition \cite{Aberg-Superposition} of a finite number of linear independent states which are not necessarily orthogonal was introduced in Ref.~\cite{Plenio-RTofS}. In addition to determining the quantum operations which do not create superposition, a partial order structure was also uncovered for the probabilistic transformations of pure superposition states \cite{Plenio-RTofS}. However, more research is needed to further explore the order relations between superposition states. On the other hand, it is a hugely challenging task for superposition, especially for those states which are maximal, due to the nonorthogonality of basis states. It is important to note that the (non)existence of a maximal state for a given set of linearly independent basis states does not necessarily imply the same for a different set. Therefore, any attempt to determine maximal states for superposition must take into account all possible linearly independent basis states, i.e., all inner product settings.

The goal of this work is to develop a structural way which comprehensively addresses the maximal states of the resource theory of superposition (RTS). To this end, we show that the Gram matrix has a crucial detail, hidden in its eigensystem, for the realization of the potential candidates for maximal states. Namely, by contemplating the eigenvalues \(\{\lambda_i\}\) and their corresponding eigenvectors \(\{\ket{\lambda_i}\}\) of the Gram matrix, we show that the eigenvector corresponding to the minimum eigenvalue \(\lambda_{\min}\) mirrors the maximal states. We illustrate our results by examining two- and three-dimensional systems together with the general case. Note that the manipulation of superposition states was studied recently in Ref.~\cite{torun2020resource}. There the maximal states for superposition were also discussed; however, the study considered a subset of initial and final state pairs, that is, states with real components and real inner products. Here we deal with this problem in its entirety by means of the Gram matrix for a general case that involves all possible inner product settings. We believe that our work bridges the gap between the resource theory of coherence \cite{Baumgratz-Coherence} and the superposition \cite{Plenio-RTofS} in the context of maximal states, where both theories are equivalent in the orthonormal case of basis vectors.

Golden states are the preeminent resources that are used in various quantum information processing, guaranteeing the optimal success probability. The teleportation protocol \cite{Bennett1993} can be regarded as one of the most prominent examples that demonstrates this point clearly. In its simplest form, a two-qubit maximally entangled state \((\ket{00}+\ket{11})/{\sqrt{2}}\) is used to perfectly teleport a quantum state between two distant parties \cite{Bennett1993,2019QTeleportationLuo}. In the case where the shared state is \(a\ket{00}+b\ket{11}\) (\(|a|>|b|>0\)), the protocol can be achieved with \(p=2|b|^2<1\) probability (i.e., not perfectly). Of course, examples that can be given in this context are not limited to teleportation. Examples of the tasks include quantum metrology \cite{Toth2014QM}, quantum communication \cite{Pirandola2017}, and cryptography \cite{Pereira2021Cryp}. Moreover, maximal states can form an effective bridge between different types of resource theories. For instance, it was shown \cite{Streltsov2018MaxC} that any amount of purity can be converted into coherence via a suitable unitary operation; therefore, purity can be identified as the maximal coherence, concluding that the resource theories of coherence and purity are closely connected. All these studies (and many more) support the priority of a maximal state among all resource states indeed and motivate further studies.

Here we look at the existence of golden states in RTS and reveal that a unique state of this type,  which is a much stronger property, exists for any dimension \(d \geq 2\) for various inner product settings. It should be said that our proposed \(d\)-dimensional maximal superposition state(s) can be used to generate all other \(d\)-dimensional states deterministically by means of superposition-free operations. A recent study by Liu \textit{et al.} \cite{Liu-OneShotRT} characterized the resourcefulness of quantum states defined in general QRTs in terms of direct one-shot resource conversion. The authors introduced the notion of max-resource state (in the strongest sense), namely \textit{root state}, which is relevant to the problem we are interested in; by definition, a root state must achieve the maximum value of any resource monotone in dimension \(d\) \cite{Liu-OneShotRT}. By discussing the constant trace condition introduced in Ref.~\cite{Liu-OneShotRT} as the overlap between the maximal state and free states, we show that it can be derived in terms of the minimum eigenvalue \(\lambda_{\min}\) of the constructed Gram matrix. We discuss our results with explicit examples.

The paper proceeds as follows. We begin in Sec.~\ref{Sec:Definitions} by reviewing the basics of the resource theory of superposition. Moreover, we introduce the Gram matrix with its relevance to the superposition of a finite number of linear independent states for arbitrary finite dimension and finalize this section by defining a map \(\Phi(\cdot)\) by means of certain types of superposition-free Kraus operators. In Sec.~\ref{Sec:ANecCon} we provide a necessary condition with respect to the map \(\Phi(\cdot)\) and discuss the possibility of the existence of maximal states for an arbitrary inner product setting. Section \ref{Sec:MaxRStates} details our results by examining particular cases together with the general case. Section \ref{Sec:Conc} summarizes our work.


\section{Definitions and Preliminaries}\label{Sec:Definitions}

In this section, we review basic definitions of the resource theory of superposition and give further useful information and tools needed to construct the rest of the paper.

The general structure of QRTs \cite{Chitambar-QRTs} is shaped around three notions: {\it free states} (whose set is denoted by \(\mathcal{F}\)), {\it resource states} (\(\mathcal{R}\)), and {\it free operations} (\(\mathcal{FO}\)). Essentially, resource states \(\mathcal{R}\) cannot be created from a set of free states \(\mathcal{F}\) under the actions of free operations \(\mathcal{FO}\) of the given resource theory. To represent the states in RTS \cite{Plenio-RTofS}, suppose \(\{\ket{c_k}\}_{k=1}^{d}\) is a normalized, linearly independent, and not necessarily orthogonal basis of the Hilbert space represented by \(\mathds{C}^{d}\), \(d \in \mathds{N}\), where their inner products (i.e., overlaps) \(\braket{c_i}{c_j}=s_{ij} \in \mathds{C}\) for \(i \neq j\) in general. Then, states defined as \(\rho = \sum_k \rho_k \ketbra{c_k}\) are called superposition-free where \(\rho_k \geq 0\) form a probability distribution.
All density operators which are not an element of \(\mathcal{F}\) are called superposition states and form the set of resource states \(\mathcal{R}\), where a pure superposition state is given by
\begin{equation}\label{eq-ch2:superposition-state}
\ket{\psi} = \sum_{k=1}^d \psi_k\ket{c_k}, \quad \psi_k\in \mathds{C}.
\end{equation}
Naturally, to be able to discuss superposition, more than one \(\psi_k\) in Eq.~\eqref{eq-ch2:superposition-state} should be nonzero. After having defined the free states and resource states, the next step is to define the free operations \(\mathcal{FO}\) for the RTS \cite{Plenio-RTofS}. A Kraus operator $K_n$ is called superposition-free if $K_n \rho K_n^{\dagger}$ $\in$ $\mathcal{F}$  for all $\rho \in \mathcal{F}$ and is of the form
\begin{eqnarray}\label{Kraus-free-n}
	K_n=\sum_{k} c_{k,n} {\ket{c_{f_{n}(k)}}\bra{c_k^{\perp}}},
\end{eqnarray}
where $c_{k,n}\in \mathds{C}$, ${f_{n}(k)}$ are arbitrary index functions~\cite{Plenio-RTofS}, and $\braket{c_i^{\perp}}{c_j}=\delta_{ij}$ where the vectors $\ket{c_k^{\perp}}$ are not normalized. A quantum operation \(\Phi(\cdot)\) is called superposition-free if it is trace preserving and can be written such that \(\Phi(\rho) = \sum_{i} K_i \rho K_i^{\dagger}\), where all \(K_i\) are free.

For a set of basis states \(\{\ket{c_k}\}_{k=1}^{d}\), an inner product setting can be represented in terms of the Gram matrix. The Gram matrix is a positive-definite matrix whose entries correspond to the inner products of the basis states, that is \(G_{ij} = \braket{c_i}{c_j}\) \ \cite{Horn-GramMatrix}. We note that Gram matrices are in general positive semidefinite; however, the Gram matrix is positive definite for the current problem since the given basis states are linearly independent. The inner product of two states, e.g., the state given in Eq.~\eqref{eq-ch2:superposition-state} and the state \(\ket{\phi} = \sum_{k=1}^{d} \phi_k\ket{c_k}\) (\(\phi_k \in \mathds{C}\)), reads
\begin{equation}\label{eq:dlevel-inner}
   \braket{\phi}{\psi} = \sum_{i=1}^{d}\sum_{j=1}^d \phi_i^{*} \psi_j \braket{c_i}{c_j}.
\end{equation}
The above equation can also be expressed in terms of the Gram matrix. By defining column vectors
\begin{equation}\label{vector-forms}
\vec{\psi} \coloneqq {\left(\psi_1, \dots, \psi_d \right)}^{\intercal}, \quad
\vec{\phi} \coloneqq {\left(\phi_1, \dots, \phi_d \right)}^{\intercal},
\end{equation}
for the states \(\ket{\psi}\) and \(\ket{\phi}\), respectively, then the inner product \(\braket{\phi}{\psi}\) given in Eq.~\eqref{eq:dlevel-inner} can be written as
\begin{equation}
    \braket{\phi}{\psi} = \vec{\phi}^{\dagger} G \vec{\psi}.
\end{equation}
Obviously, the normalization condition equals \( \vec{\psi}^{\dagger} G \vec{\psi} = 1\) for an arbitrary vector \(\vec{\psi}\). This representation has the advantage of embedding all the information about the basis states into the Gram matrix. Since a Gram matrix \(G\) is a positive-definite matrix, the normalization condition represents a (hyper)-ellipsoid equation and reduces to a (hyper)sphere in the orthonormal limit.

The following theorem --- \textit{Rayleigh Quotient} --- can be used to determine the upper and lower bounds of a quadratic equation under elliptic constraint. We will use it to prove the necessary condition proposed in Sec.~\ref{Sec:ANecCon}.

Let \(\vec{v}\) and \(\vec{u}\) be eigenvectors of a Gram matrix \(G\) corresponding to the minimum eigenvalue \(\lambda_{\min}\) and maximum eigenvalue \(\lambda_{\max}\), respectively. Let \(\vec{x}\) be an arbitrary vector under elliptic constraint represented by \(\vec{x}^{\dagger} G \vec{x} = 1\). Then, for these three vectors \(\{\vec{v}, \vec{u}, \vec{x}\}\) the Rayleigh quotient guarantees that the following inequality exists:
\begin{equation}\label{Rmin-max}
R(G, \vec{v}) = \lambda_{\min} \leq R(G, \vec{x}) \leq R(G, \vec{u}) = \lambda_{\max}.
\end{equation}
Here \(R(G, \vec{x}) = \vec{x}^{\dagger} G \vec{x}/\vec{x}^{\dagger}\vec{x}\) is the Rayleigh quotient for \(G\) and \(\vec{x}\). Then, a quadratic expression in the form \(\vec{x}^{\dagger} \vec{x}\) has an upper and a lower bound given such that
\begin{equation}\label{eq:quadratic-interval}
    \frac{1}{\lambda_{\max}} \leq \sum_i \abs{x_i}^2 \leq \frac{1}{\lambda_{\min}}
\end{equation}
(see Theorem 4.2.2 in Chap. 4 of Ref.~\cite{Horn-GramMatrix}).
Using the fact that \(\lambda_{\max} \geq \lambda_{\min} >0\), one can obtain
\begin{equation}\label{eq:quadratic-identity}
    1 - \frac{\lambda_{\min}}{\lambda_{\max}} \geq 1 - \lambda_{\min}\sum_i \abs{x_i}^2 \geq 0.
\end{equation}

We will use two sets of superposition-free Kraus operators, where they play a part in searching for maximal states as follows. Each operator in the first set \(S_1\) transforms an initial state \(\ket{\psi}\) into another state \(\ket{\phi}\) with a probability \(p_n\geq 0\) (\(\sum_{n}p_n=1\)), that is, \(K_n\ket{\psi} = \sqrt{p_n}\ket{\phi}\). On the other hand, each operator in the second set \(S_2\) maps the initial state to zero, that is, \(K_m\ket{\psi} = 0\). Then, the map can be written such that
\begin{equation}\label{eq:supfreemap-implicit}
    \Phi(\rho) = \sum_{S_1 + S_2} K_i \rho K_i^{\dagger}=\ket{\phi}\bra{\phi},
\end{equation}
where \(\rho=\ket{\psi}\bra{\psi}\).
Also, when we consider mixed states, we will use convex combinations of the map \(\Phi (\cdot)\) given by Eq.~\eqref{eq:supfreemap-implicit}, i.e., \(\sum_k p_k \Phi_k(\cdot)\). In the following, we will introduce the explicit forms of the operators in \(S_1\) and \(S_2\).

For understandable reasons, we will do (most parts of) our calculations with a vector picture.
Now consider an arbitrary superposition-free Kraus operator in the form of \eqref{Kraus-free-n} acting on a state \eqref{eq-ch2:superposition-state} such that
\begin{equation}
\begin{split}
K_n\ket{\psi} &= \Big(\sum_{k=1}^{d} c_{k,n} {\ket{c_{f_{n}(k)}}\bra{c_k^{\perp}}}\Big) \sum_{i=1}^d \psi_i\ket{c_i} \\ &= c_{1,n} \psi_1 \ket{c_{f_n(1)}} + \cdots + c_{d,n} \psi_d \ket{c_{f_n(d)}}  \\
&= \sqrt{p_n} \sum_{k=1}^d \phi_k\ket{c_k},
\end{split}
\end{equation}
where \(p_n \geq 0\). Since \(\braket{c_i^{\perp}}{c_j} = \delta_{ij}\),  the equation above can be written as a matrix multiplication. By this we mean that
\begin{equation}
(\tilde{K}_n \vec{\psi})_{i}=\sum_{j=1}^{d} (\tilde{K}_n)_{ij}\psi_j = p_n \phi_i.
\end{equation}
Here, \((\tilde{K}_n)_{ij}\) is the \((i,j)\)th entry of the matrix \(\tilde{K}_n\) and \(\phi_i\) is the \(i\)th element of the vector \(\vec{\phi}\) defined in Eq.~\eqref{vector-forms}. Note that the matrix \(\tilde{K}_n\) has at most one nonzero element in each column for all \(n\). For instance, let \(\{f_1(k)\}_{k=1}^{3} = \{1, 2, 3\}\) for a three-level system; then we have
\begin{equation}
    \tilde{K}_1
    = p_1
    \begin{pmatrix}
       \phi_1/\psi_1 & 0 & 0 \\
        0 & \phi_2/\psi_2 & 0 \\
        0 & 0 & \phi_3/\psi_3 \\
    \end{pmatrix}.
\end{equation}
Obviously, the matrix itself is in the form of an incoherent Kraus operator --- a class of free operations in the resource theory of coherence (RTC) \cite{Baumgratz-Coherence}. On the other hand, it can only be treated as an incoherent Kraus operator when it is provided that \(G = \mathds{1}\), that is, in the orthonormal case. Finally, since we have \(\sum_n \bra{\psi}K_n^{\dagger}K_n\ket{\psi} =  \sum_n (\vec{\psi}^{\dagger} \tilde{K}_n^{\dagger}) G (\tilde{K}_n \vec{\psi}) = \vec{\psi}^{\dagger} G \vec{\psi} = 1\) for all \(\vec{\psi}\), the condition to be a trace-preserving operation can be stated as
\begin{equation}
    \sum_n \tilde{K}_n^{\dagger} G \tilde{K}_n = G.
\end{equation}

Having introduced the matrix picture above, we can now define the operators in the sets \(S_1\) and \(S_2\). The operators in the set \(S_1\) are in the form
\begin{equation}\label{eq:firstsetop}
    \tilde{K}_n = P_n A_n
\end{equation}
for \(n=1, 2, \dots, d!\). Here, \(P_n\) is a permutation matrix and \(A_n\) is a square diagonal matrix.
Since \(A_n\) is a square diagonal matrix, the operators \(\tilde{K}_n\) in \(S_1\) are permutation matrices whose nonzero entries are replaced by the elements of \(A_n\). The operators in the set \(S_2\) are matrices whose rows are zero except the \(m\)th one, that is, \((\tilde{F}_m)_{ij}=(\tilde{F}_m)_{ij} \delta_{mi}\) for \(m=1, 2, \dots, d\). Then, the map \(\Phi(\cdot)\) defined in Eq.~\eqref{eq:supfreemap-implicit} involves \(d!+d\) superposition-free Kraus operators.

We close this section by using the following vector from Ref.~\cite{torun2020resource} to express the majorization condition in the following section and for future use:
\begin{equation}\label{eq:tildevector}
    \tilde{x} \coloneqq \text{diag}(x_1^{*}, \cdots, x_n^{*}) G {\left(x_1, \cdots, x_n\right)}^{\intercal}= {\left(\tilde{x}_1, \cdots, \tilde{x}_n\right)}^{\intercal},
\end{equation}
where \(\text{diag}(\cdot)\) is a diagonal matrix. Note that when \(G \rightarrow \mathds{1}\), i.e., in the orthonormal case, one has \(\tilde{\psi} = {(\abs{\psi_1}^2, \ldots, \abs{\psi_d}^2)}^{\intercal} \) and \(\tilde{\phi} = {(\abs{\phi_1}^2, \ldots, \abs{\phi_d}^2)}^{\intercal} \), which are the vectors subject to the majorization condition in RTC \cite{Nielsen-Maj,Du-CoherenceMaj}.

\section{Gram Matrix and Golden States}\label{Sec:ANecCon}

In this section, we show that a Gram matrix constructed by inner products between linearly independent basis states may shed light on the investigation of golden states. To be more precise, we uncover that the eigenvalues and the corresponding eigenvectors of a given Gram matrix and maximal states have an intrinsic connection, from which we propose a necessary condition for the existence of maximal superposition states. It is on this criterion that we concentrate our research on golden states.

Different methods enable us to examine whether there are maximal states in a given resource theory. One possible way of establishing the existence of maximal states is to employ the conditions on state transformations. The meaningful task of (superposition) state conversion aims to transform a state \(\ket{\psi}=\sum_{k=1}^d \psi_k\ket{c_k}\) into another state \(\ket{\phi}=\sum_{k=1}^d \phi_k\ket{c_k}\) via superposition-free operations, which can be achieved either deterministically or probabilistically.
The former case has been discussed in Ref.~\cite{torun2020resource}, and it has been shown that the reconstructed majorization conditions \cite{Bhatia} in a setting where all components of pure superposition states (\(\psi_k\) and \(\phi_k\)) and inner products \(\braket{c_i}{c_j}=s_{ij}\) of basis vectors are real can be expressed utilizing the Gram matrix.
Then the majorization conditions reads \(\tilde{\psi} = D \tilde{\phi}\) \cite{torun2020resource} where \(D\) is a doubly stochastic matrix [recall Eq.~\eqref{eq:tildevector}].

In the most general setting, the components of superposition states and inner products of the basis states are complex. In this case, the elements of the vectors \(\tilde{\psi}=(\tilde{\psi}_1, \dots, \tilde{\psi}_d)^{\intercal}\) (with \(\sum_{i=1}^{d}\tilde{\psi}_i = 1\)) and \(\tilde{\phi}=(\tilde{\phi}_1, \dots, \tilde{\phi}_d)^{\intercal}\) (with \(\sum_{i=1}^{d}\tilde{\phi}_i = 1\)) subject to the doubly stochastic matrix take complex values as well. Then, one cannot write elements of \(\tilde{\psi}\) and \(\tilde{\phi}\) in decreasing (or increasing) order to test whether the equation \(\tilde{\psi} = D \tilde{\phi}\) is satisfied or not.
Again, we stress that \(\tilde{\psi} = D \tilde{\phi}\) must hold \cite{torun2020resource} in the most general setting. Since it is particularly complicated to study the maximal states in the most general setting due to the reason stated above, one may work with superposition states having real components and real inner product values to fully utilize the majorization theory as in the case of Ref.~\cite{torun2020resource}. However, this means searching for maximal states in restricted regions, which has the possibility of overlooking possible maximal state candidates. Thus, for a comprehensive solution, one simply needs to resolve the most general setting.

Here, we implement a workaround to this problem in the light of the following way.
If one can build a doubly stochastic matrix \(D\) such that all its entries are \(1/d\), that is, \(D_{ij} = 1/d\), then it is clear that such a doubly stochastic matrix maps all vectors, including the ones with complex components, to \({\left(1/d, \hdots, 1/d\right)}^{\intercal}\). Therefore, the equation \(\tilde{\psi} = D \tilde{\phi}\) always holds given that the initial state satisfies \(\tilde{\psi}_i = 1/d\) for all \(i\) for an arbitrary final state. One point to be noted, the additional constraint, namely, the condition on the completeness \cite{torun2020resource}, adds a restriction to the set of inner product settings. As a necessary condition for a maximal state in the cases where the transformation is carried out by the map \(\Phi(\cdot)\) given in Eq.~\eqref{eq:supfreemap-implicit}, we have the following.


\begin{proposition}\label{Prop:1}
 Let \(G\) be a Gram matrix that represents the inner product settings of superposition constructed by a set of linearly independent basis states \(\{\ket{c_i}\}_{i=1}^{d}\). Then a maximal state \(\vec{\psi}\) is necessarily an eigenvector of \(G\) corresponding to the minimum eigenvalue \(\lambda_{\min}\) and the components of \(\vec{\psi}\) must satisfy \(\tilde{\psi}_i = 1/d\) for all \(i =1, 2, \dots, d\).
\end{proposition}


\begin{corollary}\label{Corollary:1}
The form of the (candidate) maximal state dictated by Proposition \ref{Prop:1} is then as follows
\begin{equation}\label{eq:maximalform-vec}
    \vec{\psi} = \sqrt{\frac{1}{d\lambda_{\min}}} {\left(e^{i\theta_1}, \dots, e^{i\theta_d}\right)}^{\intercal},
\end{equation}
and equivalently in the ``Dirac'' notation
\begin{equation}\label{eq:maximalform-braket}
    \ket{\psi} = \sqrt{\frac{1}{d\lambda_{\min}}} \sum_{j=1}^{d} {e^{i\theta_j}} \ket{c_j}.
\end{equation}
\end{corollary}
As expected, the state given in Eq.~\eqref{eq:maximalform-braket} reduces to the maximal state of coherence \cite{Baumgratz-Coherence}, which is given by
\begin{equation}\label{eq:coherence-maximal}
    \ket{\Psi_d} = \frac{1}{\sqrt{d}} \sum_{j=1}^{d} e^{i\tau_j}\ket{j},
\end{equation}
where \(\{\ket{j}\}_{j=1}^d\) forms an orthonormal basis and \(\tau_j \in \mathds{R}\). In the orthonormal limit, we have \(G \rightarrow \mathds{1}\), and therefore \(\lambda_{\min} \rightarrow 1\). We note that in the RTC, one has the freedom to choose \(\tau_i\) arbitrarily since the states are equivalent to each other by means of incoherent unitary operations \cite{Baumgratz-Coherence}. However, relative phases of the maximal states of RTS are fixed due to the fact that operations acting on relative phases are not free \cite{Plenio-RTofS}.

Proposition \ref{Prop:1} is a compound statement where the first part establishes a relation between the Gram matrix (hence the given basis vectors) and the maximal states. The second part introduces the doubly stochastic matrix picture of the majorization theory \cite{Bhatia}. Therefore, the combination of those two parts provides a unified statement regarding the maximal states of the RTC \cite{Baumgratz-Coherence} and RTS \cite{Plenio-RTofS}.

\begin{proof}
Given an initial state \(\vec{\psi}=(e^{i\theta_1}\psi_1, \dots, e^{i\theta_d}\psi_d)^{\intercal}\)
and a final state \(\vec{\phi}=(e^{i\gamma_1}\phi_1, \dots, e^{i\gamma_d}\phi_d)^{\intercal}\),
where \(\theta_j , \gamma_j \in [0, 2\pi)\) and \(\psi_j , \phi_j \in \mathds{R}^{\geq 0}\), the transformation from the initial state to the final state is carried out by the operators in the first set \(S_1\) such that  \(\tilde{K}_n\vec{\psi} = \sqrt{p_n}\vec{\phi}\). Here, the probabilities are chosen to be \(p_n = 1/d!\) for \(n=1, 2, \dots, d!\) and the Kraus operators \(\tilde{K}_n\) have to be in the form (defining \(\epsilon_{ij} \coloneqq \gamma_i - \theta_j\))
\begin{equation}\label{eq:dlevel-rankd-kraus}
    {(\tilde{K}_{n})}_{ij} = \sqrt{\frac{1}{d!}} e^{i(\epsilon_{ij, n})} \frac{\phi_{i,n}}{\psi_{j,n}}
\end{equation}
for \(i,j=1, 2, \dots, d\) and \(n=1, 2, \dots, d!\). From Eq.~\eqref{eq:firstsetop}, it is clear that whether an entry of the matrix is zero or not depends on the index \(n\). The operators in the second set \(S_2\) are defined as \((\tilde{F}_m)_{ij}=(\tilde{F}_m)_{ij} \delta_{mi}\) for \(m=1, 2, \dots, d\).

Now, in the first part of the proof of Proposition \ref{Prop:1}, we show that a maximal state has to be an eigenvector of a Gram matrix corresponding to the minimum eigenvalue. Assume that \(\vec{\psi}\) is a maximal state and the map \(\Phi(\cdot)\) given in Eq.~\eqref{eq:supfreemap-implicit} is a superposition-free operation. Then, the Kraus operators belonging to the sets \(S_1\) and \(S_2\) must satisfy
\begin{equation}\label{eq:tracecon}
\sum_{n=1}^{d!} \tilde{K}_n^{\dagger} G \tilde{K}_n +\sum_{m=1}^{d}\tilde{F}_m^{\dagger} G \tilde{F}_m = G.
\end{equation}
We express Eq.~\eqref{eq:tracecon} in index notation to calculate the entries of the resulting matrix of the left-hand side. First, we focus on the operators in the set \(S_1\), that is
\begin{equation}\label{Set1KrausMap}
    {\bigg(\sum_{n=1}^{d!}\tilde{K}_n^{\dagger} G \tilde{K}_n\bigg)}_{il} = \sum_{n=1}^{d!} \left[ \sum_{j=1}^{d} \sum_{k=1}^{d} {(\tilde{K}^{\dagger}_n)}_{ij} G_{jk} {(\tilde{K}_n)}_{kl}  \right],
\end{equation}
for \(i,l=1, 2, \dots, d\). To shorten the expressions let us define \(\mathds{K} \coloneqq \sum_{n=1}^{d!}\tilde{K}_n^{\dagger} G \tilde{K}_n\). Combining Eq.~\eqref{eq:dlevel-rankd-kraus} with Eq.~\eqref{Set1KrausMap} using the fact that \({(\tilde{K}_n^{\dagger})}_{ij} = {(\tilde{K}_n^{*})}_{ji}\), the first \(d!\) terms read
\begin{equation}\label{Yil-1st}
    {\mathds{K}}_{il} = \sum_{n=1}^{d!} \left[\sum_{j=1}^{d} \sum_{k=1}^{d}\frac{1}{d!} e^{i(\epsilon_{kl, n}-\epsilon_{ji, n})} G_{jk} \frac{\phi_{j,n}}{\psi_{i,n}}\frac{\phi_{k, n}}{\psi_{l, n}} \right].
\end{equation}
Using the fact that \(\sum_j\sum_k a_j b_k = \sum_j a_j b_j + \sum_{j, k, (j \neq k)} a_j b_k\), Eq.~\eqref{Yil-1st} goes one step further. There, the index \(n\) dependence of the two sum terms on the right-hand side can be calculated by invoking the counting problem. As we mentioned before, an operator \(\tilde{K}_n\) defined in Eq.~\eqref{eq:firstsetop} has at most one nonzero entry in each column and row. Therefore, for \(i = l\), the sum over the index set \(j \neq k\) is zero for all \(n\) and the term \(\phi_j/\psi_i\) repeats \(d!/d\) times for each \(j\). Hence, the diagonal entries are obtained as
\begin{equation}\label{KnSecondSet11}
    {\mathds{K}}_{ii}
    = \frac{1}{d\psi_{i}^2}\sum_{j=1}^{d} \phi_{j}^2.
\end{equation}
For \(i \neq l\) the sum over the index set \(j = k\) is zero for all \(n\) due to the exact reason stated above. Since \(i \neq l\) and \(j \neq k\), there exists a Kraus operator \(\tilde{K}_{n^{\prime}}\) with entries such that
\({(\tilde{K}_{n^{\prime}})}_{ij} =  e^{i\epsilon_{ij}}({\phi_i}/{\psi_j})\) and \({(\tilde{K}_{n^{\prime}})}_{kl} =  e^{i\epsilon_{kl}}({\phi_k}/{\psi_l})\) for an arbitrary index \(n^{\prime}\). Starting from \(n^{\prime}\), among numbers of  permutations \(d!\) of the \(n^{\prime}\)th Kraus operator, only \((d-2)!\) of those do not change the indices. Hence summing over the index \(n\) and using the shorthand notation \(\theta_{i-l} \equiv \theta_i - \theta_l\) and \(\gamma_{k-j} \equiv \gamma_k - \gamma_j\), one obtains
\begin{equation}\label{KnSecondSet12}
    {\mathds{K}}_{i \neq l} = \frac{(d-2)!}{d!} \frac{e^{i\theta_{i-l}}}{\psi_i \psi_l} \sum_{\substack{j,k \\ j \neq k}}^{d}  e^{i\gamma_{k - j}} G_{jk} \phi_{j}\phi_{k}.
\end{equation}
The second sum term of Eq.~\eqref{eq:tracecon} reads
\begin{equation}\label{FmSecondSet2}
    \begin{split}
    {\bigg(\sum_{m=1}^{d}\tilde{F}_m^{\dagger} G \tilde{F}_m \bigg)}_{il} &= \sum_{m=1}^{d} \left[\sum_{j,k}^{d} {(\tilde{F}_m)}_{ji}^{*} \delta_{mj} G_{jk} {(\tilde{F}_m)}_{kl} \delta_{mk} \right] \\ &= \sum_{m=1}^{d} {(\tilde{F}_m)}_{mi}^{*} {(\tilde{F}_m)}_{ml}.
    \end{split}
\end{equation}
After combining Eqs.~\eqref{KnSecondSet11}, \eqref{KnSecondSet12}, and \eqref{FmSecondSet2}, we obtain
\begin{equation}\label{eq:tracecon-openform}
    \begin{split}
    G_{il}= & \bigg[\frac{1}{d\psi_{i}^2}\sum_{j=1}^{d} \phi_{j}^2\bigg] \delta_{il} + \sum_{m=1}^{d} {(\tilde{F}_m)}_{mi}^{*} {(\tilde{F}_m)}_{ml} \\ &+ \bigg[\frac{(d-2)!}{d!} \frac{e^{i\theta_{i - l}}}{\psi_i \psi_l} \sum_{\substack{j,k \\ j \neq k}}^{d}  e^{i\gamma_{k - j}} G_{jk} \phi_{j}\phi_{k}\bigg] (1-\delta_{il}).
    \end{split}
\end{equation}
On the one hand, from this form of Eq.~\eqref{eq:tracecon-openform} it is arduous to reveal the form of the maximal state(s). On the other hand, we can employ Eq.~\eqref{eq:tracecon-openform} to obtain a condition so that we can invoke the \textit{Rayleigh Quotient} theorem introduced in Sec.~\ref{Sec:Definitions}. Now consider the case \(l = i\) in \eqref{eq:tracecon-openform}, which gives
\begin{equation}\label{psii-phij-elements}
\frac{1}{d}\sum_{j=1}^{d} \phi_{j}^2 + \sum_{m=1}^{d} \abs{{(\tilde{F}_m)}_{mi}}^{2}\psi_{i}^2 = \psi_{i}^2,
\end{equation}
for \(i=1, 2, \dots, d\). By summing over the index \(i\), one can obtain the inequality
\begin{equation}
    \sum_{i=1}^{d}\psi_i^2 \geq \sum_{j=1}^{d}\phi_j^2.
\end{equation}
It is obvious that a pair of states that do not provide this relation cannot satisfy each component of Eq.~\eqref{psii-phij-elements} hence Eq.~\eqref{eq:tracecon} simultaneously. The expressions on both sides are quadratic expressions in the same form constrained by the normalization condition; therefore, they have an upper and a lower bound. For a maximal state, the left-hand side has to take the maximum value so that it always holds for an arbitrary final state. By invoking the \textit{Rayleigh Quotient}, one concludes that a maximal state is necessarily an eigenvector of the Gram matrix corresponding to the minimum eigenvalue. This completes the first part of the proof.

To provide the second part of the proof of Proposition \ref{Prop:1}, we return to Eq.~\eqref{eq:tracecon-openform}. Since there is no sum term over the indices \(i\) and \(l\), we can multiply Eq.~\eqref{eq:tracecon-openform} by
\(e^{-i\theta_{i-l}}\psi_i\psi_l\). Then, using the fact that \({e^{-i\theta_{i-l}}}(\psi_i\psi_l/\psi_i^2)\delta_{il} = \delta_{il}\) and summing over the index \(l\), Eq.~\eqref{eq:tracecon-openform} becomes
\begin{equation}\label{1dprooflaststep}
\begin{split}
    \sum_{l=1}^{d} \frac{G_{il} \psi_i \psi_l}{e^{i\theta_{i-l}}} =&
    \bigg[\frac{1}{d} \sum_{j=1}^{d} \sum_{l=1}^{d} \phi_{j}^2\bigg] \delta_{il} \\
    & + \bigg[\frac{(d-2)!}{d!} \sum_{l=1}^{d} \sum_{\substack{j,k \\ j \neq k}}^{d}  e^{i\gamma_{k-j}} G_{jk} \phi_{j}\phi_{k}\bigg] (1-\delta_{il}) \\
    & + \sum_{m=1}^{d} \sum_{l=1}^{d} {(\tilde{F}_m)}_{mi}^{*} {(\tilde{F}_m)}_{ml} e^{-i\theta_{i-l}}\psi_i \psi_l.
\end{split}
\end{equation}
To further simplify \eqref{1dprooflaststep}, recall that each operator in the set \(S_2\) maps the initial state to zero, that is, \(\tilde{F}_m\vec{\psi} =  \sum_{l=1}^{d} {(\tilde{F}_m)}_{ml}e^{i\theta_l}\psi_l = 0\). Thus, the third term on the right-hand side is zero. We can then write \eqref{1dprooflaststep} as
\begin{equation}
    \begin{split}
    \sum_{l=1}^{d} G_{il} e^{-i\theta_{i-l}}\psi_i \psi_l
    &= \frac{1}{d} \bigg[\sum_{j=1}^{d} \phi_{j}^2 + \sum_{\substack{j,k \\ j \neq k}}^{d}  e^{i\gamma_{k-j}} G_{jk} \phi_{j}\phi_{k}\bigg] \\
    &= \frac{1}{d} \vec{\phi}^{\dagger} G \vec{\phi}.
    \end{split}
\end{equation}
Then, it is obvious that one is required to satisfy
\begin{equation}\label{eq:dlevel-maj}
\sum_{l=1}^{d} e^{-i\theta_{i-l}} G_{il} \psi_i \psi_l = \tilde{\psi}_i = \frac{1}{d},
\end{equation}
for \(i=1, 2, \dots, d\), where the left-hand side is the vector defined in Eq.~\eqref{eq:tildevector} in index notation form. This concludes the second part of the proof of Proposition \ref{Prop:1}.
\end{proof}

In summary, we have proposed that a maximal state is necessarily an eigenvector corresponding to the minimum eigenvalue of the Gram matrix, and must satisfy Eq.~\eqref{eq:dlevel-maj}. For a clear view of how these conditions dictate the form of a maximal state, assume that \(G \vec{\psi} = \lambda_{\text{min}}\vec{\psi}\) and \(\tilde{\psi} = {(\tilde{\psi}_1, \dots,  \tilde{\psi}_d)}^{\intercal} = {(1/d, \dots, 1/d)}^{\intercal}\). Then, it follows that \(\text{diag}(\psi_1^{*}, \dots, \psi_d^{*})G \vec{\psi} = \lambda_{\text{min}} {(\abs{\psi_1}^2, \dots, \abs{\psi_d}^2)}^{\intercal}\). Therefore, we have \(\tilde{\psi}_i = \lambda_{\text{min}} \abs{\psi_i}^2\), which dictates \(\abs{\psi_i} = \abs{\psi_j}\) for all \(i, j\).

Proposition \ref{Prop:1} considers the superposition-free state transformations that are carried out by a subset of all types of superposition-free Kraus operators. Therefore, other types of operators that are not included in the set \(S_1\) and the set \(S_2\) may provide additional constraints. However, consider a state transformation for a three-dimensional system, where \(\vec{\psi}\) and \(\vec{\phi}\) are initial and final states, respectively. Assume that \(\psi_i \neq 0\) and \(\phi_i \neq 0\) for all \(i\) (i.e., both states have the same superposition rank \cite{Plenio-RTofS}, \(r_S(\vec{\psi})=r_S(\vec{\phi})=3\)). Then it is necessary to use the operators in the first set \(S_1\) in a way that \(\tilde{K}_n \vec{\psi} = \sqrt{p_n}\vec{\phi}\). For any other type of superposition-free Kraus operators, one needs \({M}_m \vec{\psi} = 0\), since (matrix) \(\rank({M}_m) < 3\) for all \({M}_m \notin S_1\). Notably, numerical calculations indicate that for an initial state that satisfies the second part of Proposition \ref{Prop:1} (that is, \(\tilde{\psi}_i = 1/3\) for all \(i\)) but violates the first part of Proposition \ref{Prop:1} (that is, \(G\vec{\psi} = \lambda_{\min}\vec{\psi}\)), one will have \(\sum_{S_1} \tilde{K}_n^{\dagger}G\tilde{K}_n > G\). Then it is not possible to obtain a 
trace-preserving operation by introducing additional operators. Moreover, in the following section we will show that the states that satisfy Proposition \ref{Prop:1} are valid maximal states. Thus,  considering the above reasoning and the results below, we conjecture that Proposition \ref{Prop:1} is a necessary and sufficient condition.

Finally, let us look at the \(l_1\)-norm and the relative entropy of coherence \cite{Baumgratz-Coherence} and superposition \cite{Plenio-RTofS}. It is easy to see that the \(l_1\)-norm of the maximal coherent state \(\ket{\Psi_{d}}=({1}/{\sqrt{d}}) \sum_{i=1}^{d}\ket{i}\) is equal to \(d-1\) and the relative entropy of coherence, which is given by \(C_{\text{rel.ent.}}(\rho) = \min_{\sigma \in \mathcal{F}} S(\rho||\sigma)\), is equal to \(\ln(d)\) for \(\ket{\Psi_{d}}\). Note that \(C_{\text{rel.ent.}}(\rho) = S(\rho_{\text{diag}}) - S(\rho)\), where \(S\) is the von Neumann entropy and \(\rho_{\text{diag}}\) denotes the state obtained from \(\rho\) by deleting all off-diagonal elements \cite{Baumgratz-Coherence}. When we consider the superposition, we report that the upper bounds for the \(l_1\)-norm and the relative entropy are \((d-1)/\lambda_{\min}\) and \(\ln(d/\lambda_{\min})\), respectively. These upper bounds are attained for the maximal states introduced in the following section. As expected, \(\lambda_{\min}=1\) in the orthonormal limit of basis vectors and thus coherence and superposition have the same results.

\section{maximal superposition states}\label{Sec:MaxRStates}

\subsection{Qubit systems}

The previous works \cite{Plenio-RTofS,torun2020resource} regarding the existence of the maximal states of a qubit system considered only real inner product settings, that is, \(\braket{c_1}{c_2} \in \mathds{R}\). However, the following proposition encapsulates all possible inner product settings including complex-valued inner products, which therefore combines and generalizes the previous works \cite{Plenio-RTofS,torun2020resource}.
\begin{proposition}\label{Prop:2}
    Let \(G\) be a Gram matrix that represents the inner product settings of superposition constructed by a set of linearly independent basis states \(\{\ket{c_1}, \ket{c_2}\}\). The eigenvector \(\vec{\psi}\) corresponding to the minimum eigenvalue \(\lambda_{\min}\) is then a maximal state.
\end{proposition}

\begin{proof}
    For a proof of Proposition \ref{Prop:2}, see the proof of Proposition \ref{Prop:3}.
\end{proof}

\begin{figure}[b]
\centering
\includegraphics[width=1\columnwidth]{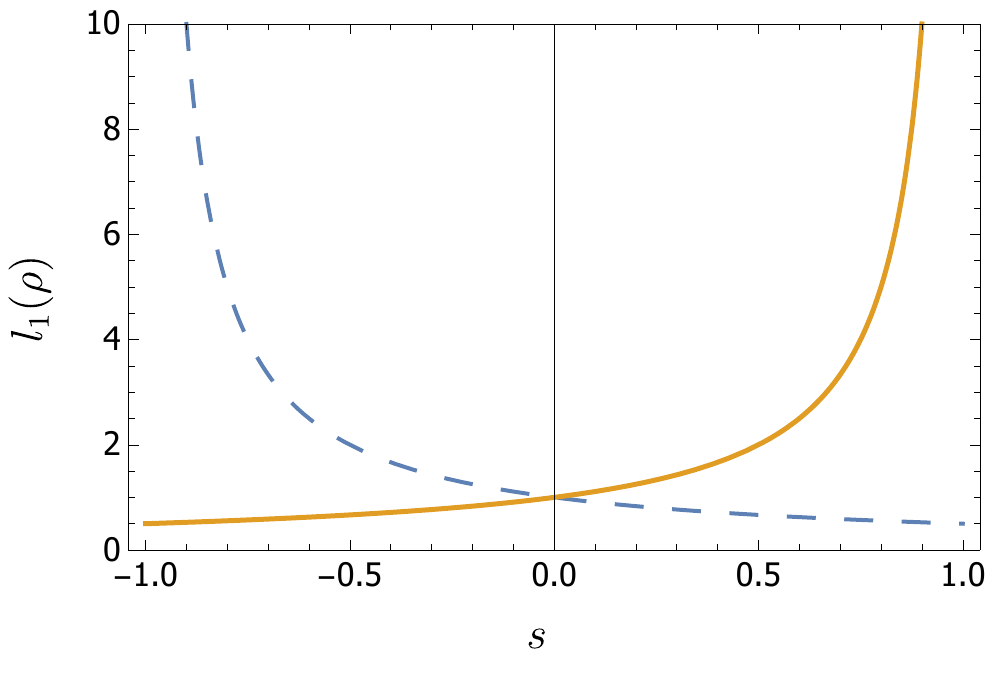}
\caption{Plots of the \(l_1\) norm of superposition of the states \(\ket{\Psi_{2}^{-}}\) and \(\ket{\Psi_{2}^{+}}\) given in Eqs.~\eqref{eq:twolevel-maximal-negative} and \eqref{eq:twolevel-maximal-positive}.
The \(l_1\) norm of superposition of the state \(\ket{\Psi_{2}^{-}}\) is \(l_1(\rho_{\ket{\Psi_{2}^{-}}})={1}/{(1-s)}\) (the solid orange line) and the \(l_1\) norm of superposition of the state \(\ket{\Psi_{2}^{+}}\) is \(l_1(\rho_{\ket{\Psi_{2}^{+}}})={1}/{(1+s)}\) (the dashed blue line). We also recall that the \(l_1\) norm of coherence of the maximally coherent state \(\ket{\Psi_{2}}=(\ket{0}+\ket{1})/\sqrt{2}\) is equal to  \({1}\) for \(d=2\), where \(l_1(\rho_{\ket{\Psi_{2}^{-}}})=l_1(\rho_{\ket{\Psi_{2}^{+}}})=1\) for \(s=0\). For any given state \(\ket{\phi}=\phi_1\ket{c_1}+\phi_2\ket{c_2}\) with \(\braket{c_1}{c_2}=s\) in dimension two, we have the following results: \(l_1(\rho_{\ket{\Psi_{2}^{-}}}) > l_1(\rho_{\ket{\phi}})\) for \(s\in[0, 1)\) and  \(l_1(\rho_{\ket{\Psi_{2}^{+}}}) > l_1(\rho_{\ket{\phi}})\) for \(s\in(-1, 0]\).}
\label{Fig1:2D}
\end{figure}

We first deal with the known examples. Let us consider the case  where the inner product of the basis states is chosen to be real, that is, \(\braket{c_1}{c_2}=s\in(-1, 1)\). The Gram matrix then has two eigenvalues \(\lambda_1=1-s\) and \(\lambda_2=1+s\) with the corresponding eigenvectors \( \vec{x}_1 = {(1, -1)}^{\intercal} \) and  \( \vec{x}_2 = {(1, 1)}^{\intercal} \).
Obviously, two different situations occur: \(\lambda_{\text{min}}=1-s\) when \(s\in[0,1)\) and \(\lambda_{\text{min}}=1+s\) when \(s\in(-1,0]\). In understanding what kind of consequences appear, it is necessary to discuss the maximally resourceful superposition states separately for these two cases. Then, the state for the former case can be written as
\begin{equation}\label{eq:twolevel-maximal-negative}
\ket{\Psi_{2}^{-}} = \frac{1}{\sqrt{2(1-s)}}\Big(\ket{c_1} - \ket{c_2}\Big), \quad s\in[0, 1),
\end{equation}
after the normalization \(\vec{x}_1^{\dagger} G \vec{x}_1 = 1\). It is possible to convert \(\ket{\Psi_{2}^{-}}\)
into any two-dimensional state via superposition-free operations \cite{Plenio-RTofS,torun2020resource}, that is, the state \eqref{eq:twolevel-maximal-negative} is maximally resourceful when \(s\in[0, 1)\). In the latter case, the state after the normalization \(\vec{x}_2^{\dagger} G \vec{x}_2 = 1\) can be written as
\begin{equation}\label{eq:twolevel-maximal-positive}
    \ket{\Psi_{2}^{+}} = \frac{1}{\sqrt{2(1+s)}}\Big(\ket{c_1} + \ket{c_2}\Big), \quad  s\in(-1, 0].
\end{equation}
By using superposition-free operations one can transform \(\ket{\Psi_{2}^{+}}\)
into any two-dimensional state with \(s\in(-1, 0]\) \cite{torun2020resource}.
It follows from Eqs.~\eqref{eq:twolevel-maximal-negative} and \eqref{eq:twolevel-maximal-positive} that for the states \(\ket{\Psi_{2}^{-}}\) and \(\ket{\Psi_{2}^{+}}\) one can immediately obtain that \(\tilde{\psi}_1=\tilde{\psi}_2=1/2\). For a given inner product value \(s\), the \(l_1\)-norm of superposition \cite{Plenio-RTofS} reaches its maximum value for the states given in Eq.~\eqref{eq:twolevel-maximal-negative} when \(s\in[0,1)\) and Eq.\eqref{eq:twolevel-maximal-positive} when \(s\in(-1, 0]\) (see Fig.~\ref{Fig1:2D}). Note that, as expected, the states \(\ket{\Psi_{2}^{-}}\) and \(\ket{\Psi_{2}^{+}}\) given in Eqs.~\eqref{eq:twolevel-maximal-negative} and \eqref{eq:twolevel-maximal-positive} reduce to two unitarily (i.e., incoherent unitary) equivalent forms of the maximal state of coherence \cite{Baumgratz-Coherence} for \(d=2\) in the orthonormal limit (i.e., \(s=0\)).

At this point, it is worth mentioning the constant trace condition introduced in Ref.~\cite{Liu-OneShotRT}.
From a geometric point of view, a maximal state \(\ket{\Psi}\) is the state that has a constant overlap with all free states, that is,
\(\text{tr}(\ket{\Psi}\bra{\Psi}\delta) = \text{const}\) \(\forall\delta \in \mathcal{F}\) \cite{Liu-OneShotRT}.
For two-dimensional systems with two pure basis states \(\ket{c_1}\) and \(\ket{c_2}\), it is easy to
see that \(\text{tr}(\ket{\Psi_{2}^{-}}\braket{\Psi_{2}^{-}}{c_i}\bra{c_i}) = (1-s)/2\) and \(\text{tr}(\ket{\Psi_{2}^{+}}\braket{\Psi_{2}^{+}}{c_i}\bra{c_i}) = (1+s)/2\) for \(i=1,2\).
In addition, a geometrical consideration of the degree of freedom in orientation of the pure basis states (see Fig.~\ref{Fig2:BlochS}) can be useful to associate two theories, coherence and superposition.

\begin{figure}[t]
  \centering
  \includegraphics[width=.6\columnwidth]{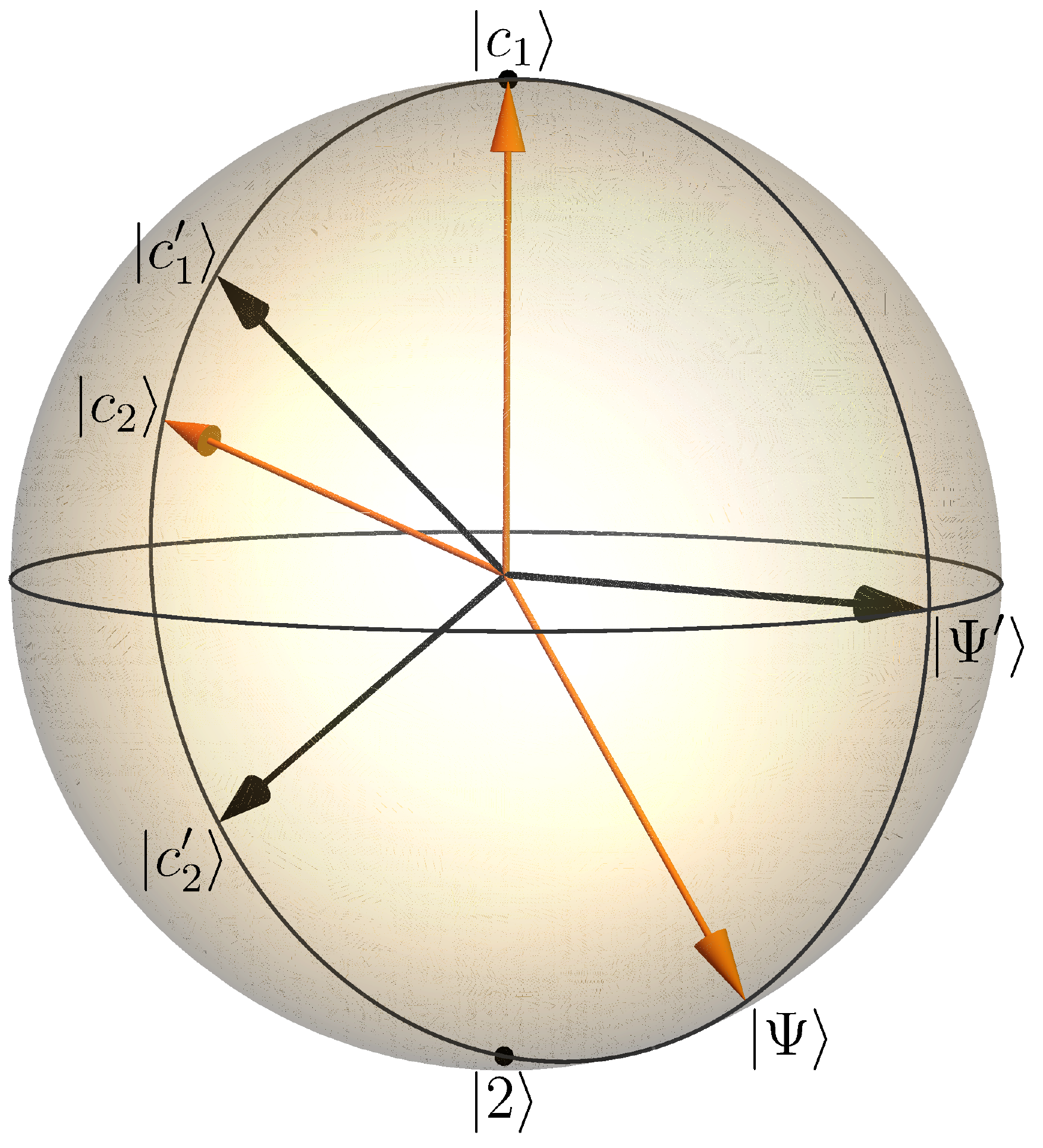}
  \caption{Representation of the degree of freedom in orientation of the pure basis states. Here, \(\ket{c_1}\) is aligned with the computational basis \(\ket{1}\) so that \(\ket{c_1} = \ket{1}\) and \(\braket{c_1}{c_2} = \braket{c^{\prime}_1}{c^{\prime}_2} = -{4}/{5}\), that is, two sets \(\{\ket{c_1}, \ket{c_2}\}\) and \(\{\ket{c^{\prime}_1}, \ket{c^{\prime}_2}\}\) correspond to the same Gram matrix. First recall that \(G = V^{\dagger}V\), where \(V\ket{i}= \ket{c_i}\). Then, for a given Gram matrix \(G\), one can have \(G = V^{\dagger}(U^{\dagger}U)V\), where \(U\) is a unitary matrix. Therefore, one can orient the basis states and hence the maximal state arbitrarily by using the unitary \(U\). In this way, the maximal state of coherence and of superposition may coincide, e.g., \(\ket{\Psi'} = \sqrt{\frac{5}{2}}\left(\ket{c^{\prime}_1}+\ket{c^{\prime}_2}\right)\) for superposition and \(\ket{\Psi'} = \frac{1}{\sqrt{2}}\left(\ket{1}+\ket{2}\right)\) for coherence.}
  \label{Fig2:BlochS}
\end{figure}

Before closing this section, we consider the inner product setting \(\braket{c_1}{c_2}=s e^{i\theta}\), where \(s\in(-1, 1)\) and \(\theta\in[0, \pi]\), which represents the most general case. Then, one can show that the Gram matrix has two eigenvalues \(\lambda_1=1-s\) and \(\lambda_2=1+s\) with the corresponding eigenvectors \( \vec{x}_1 = {(1, -e^{-i\theta})}^{\intercal} \) and  \( \vec{x}_2 = {(1, e^{-i\theta})}^{\intercal} \). Then we have
\begin{equation}\label{eq:d2Complexinner}
\ket{\Phi_{2}^{\mp}} = \frac{1}{\sqrt{2(1\mp s)}}\Big(\ket{c_1} \mp e^{-i\theta}\ket{c_2}\Big).
\end{equation}
Note that the phase of the inner product of the basis states is compensated by the relative phase of the maximal state given by Eq.~\eqref{eq:d2Complexinner}, so that \(\tilde{\psi} = {(1/2, 1/2)}^{\intercal}\). We conclude the qubit system by stating that there exists a maximal state for every inner product setting of a qubit system, where \eqref{eq:d2Complexinner} corresponds to the general case.

\subsection{Three-dimensional systems}

Importantly, the values of inner products must satisfy an inequality arising from the linear independence of the basis states \(\ket{c_1}\), \(\ket{c_2}\), and \(\ket{c_3}\), so is the determinant of the Gram matrix~\cite{torun2020resource,Horn-GramMatrix}, which is given such that
\begin{eqnarray}
1 - |s_{12}|^2 - |s_{13}|^2  - |s_{23}|^2 + s_{12} s_{13}^{*} s_{23} + s_{12}^{*} s_{13} s_{23}^{*} > 0. \ \
\end{eqnarray}
We start our analysis with real and equal inner products \(\braket{c_1}{c_2}=\braket{c_1}{c_3}=\braket{c_2}{c_3}=s\) as in qubit systems. We then have \(\det(G) = {1- 3s^2 + 2s^3 > 0}\), which implies \(s\in(-\frac{1}{2}, 1)\). One can obtain the eigensystem of such a setting as
\begin{eqnarray}
     \lambda_1 = \lambda_2 = 1-s, \quad  \lambda_3 = 1+2s \ ;
\end{eqnarray}
\begin{eqnarray}
     \vec{x}_1 = {(1, 0, -1)}^{\intercal}, \quad \vec{x}_2 = {(1, -1, 0)}^{\intercal}, \quad \vec{x}_3 = {(1,1,1)}^{\intercal}. \ \
\end{eqnarray}
In this setting, the eigenvector corresponding to the \(\lambda_{\text{min}}\) changes whether the inner product \(s\in(-\frac{1}{2}, 1)\) is positive or negative. Here it is straightforward to show that a state in the form of \(\vec{x}_3\), which can be written as
\begin{equation}\label{GU-3dim-negative}
\ket{\Psi_{3}^{+}} = \frac{1}{\sqrt{3(1+2s)}}\Big(\ket{c_1} + \ket{c_2} + \ket{c_3}\Big), \ \ s\in(-\frac{1}{2}, 0]
\end{equation}
after the normalization \(\vec{x}_3^{\dagger} G \vec{x}_3 = 1\), is a maximal state for \(s\in(-\frac{1}{2}, 0]\). In other words, the state \(\ket{\Psi_{3}^{+}}\) in Eq.~\eqref{GU-3dim-negative}
can be transformed into any three-dimensional state via superposition-free operations. Again, it is obvious that \(\tilde{\psi}_1=\tilde{\psi}_2=\tilde{\psi}_3=1/3\) for the state \eqref{GU-3dim-negative}.

\begin{table}[t]
	\caption{Form of the maximal states for particular inner product settings of basis states \(\ket{c_1}\), \(\ket{c_2}\), and \(\ket{c_3}\), where \(\braket{c_i}{c_j}=s_{ij}\). For a given inner product setting, if \(\lambda_{\min}=1+2s\), we denote maximal states by \(\ket{\Psi_{3}^{+}}\) [one form is given in Eq.~\eqref{GU-3dim-negative}], and if \(\lambda_{\min}=1-2s\), we denote maximal states by \(\ket{\Psi_{3}^{-}}\) [one form is given in Eq.~\eqref{GU-3dim-positive}]. Examples are not limited to real inner product settings; see the example given in the last row of the table.}
	\label{Table:D3GoldenUnits}
	\centering
	\begin{ruledtabular}
		\begin{tabular}{c c c}
			
			Inner product settings: \\ \(\{s_{12}, s_{13}, s_{23}\}\)   & \(\big\{\lambda_{\min}, \vec{x} \big\}\)  & Range of \(s\)  \\ [1ex]  \hline
			\(\{s,s,s\}\)  & \(\big\{1+2s, {(1, 1, 1)}^{\intercal}\big\}\) &  \(s\in(-\frac{1}{2}, 0]\)    \\
			\(\{-s, s, s\}\)   & \(\big\{1-2s, {(1, 1, -1)}^{\intercal}\big\}\) &  \(s\in[0, \frac{1}{2})\)  \\
			\(\{s, -s, s\}\)  & \(\big\{1-2s, {(1, -1, 1)}^{\intercal}\big\}\) &  \(s\in[0, \frac{1}{2})\)  \\
			\(\{s, s, -s\}\)  & \(\big\{1-2s, {(-1, 1, 1)}^{\intercal}\big\}\) &  \(s\in[0, \frac{1}{2})\)  \\
			\(\{s, -s, -s\}\)  & \(\big\{1+2s, {(1, 1, -1)}^{\intercal}\big\}\) & \(s\in(-\frac{1}{2}, 0]\)  \\
			\(\{-s, s, -s\}\) & \(\big\{1+2s, {(1, -1, 1)}^{\intercal}\big\}\) & \(s\in(-\frac{1}{2}, 0]\)  \\
			\(\{-s, -s, s\}\) & \(\big\{1+2s, {(-1, 1, 1)}^{\intercal}\big\}\) &  \(s\in(-\frac{1}{2}, 0]\)  \\
			\(\{-s, -s, -s\}\)  & \(\big\{1-2s, {(1, 1, 1)}^{\intercal}\big\}\) &  \(s\in[0, \frac{1}{2})\)  \\
			\(\{s, is, -is\}\)  & \(\big\{1-2s, {(-i, i, 1)}^{\intercal}\big\}\) &  \(s\in[0, \frac{1}{2})\) \\
		\end{tabular}
	\end{ruledtabular}
\end{table}

Next, consider the following inner product setting: \(\braket{c_1}{c_2}=-s\) and \(\braket{c_1}{c_3}=\braket{c_2}{c_3}=s\). We then have
\(\det(G) = {1- 3s^2 - 2s^3 > 0}\) which implies \(s\in(-1, \frac{1}{2})\). One can obtain the eigensystem of such a setting as
\begin{eqnarray}
	\lambda_1 = \lambda_2 = 1+s, \quad  \lambda_3 = 1-2s \ ;
\end{eqnarray}
\begin{eqnarray}
	\vec{x}_1 = {(1, 0, -1)}^{\intercal}, \quad \vec{x}_2 = {(1, -1, 0)}^{\intercal}, \quad \vec{x}_3 = {(1, 1, -1)}^{\intercal}. \quad
\end{eqnarray}
Again, the eigenvector corresponding to the \(\lambda_{\text{min}}\) changes whether the inner product \(s\in(-1, \frac{1}{2})\) is positive or negative. Here it is straightforward to show that a state in the form of \(\vec{x}_3\), which can be written as
\begin{equation}\label{GU-3dim-positive}
	\ket{\Psi_{3}^{-}} = \frac{1}{\sqrt{3(1-2s)}}\Big(\ket{c_1} + \ket{c_2} - \ket{c_3}\Big), \ \ s\in[0, \frac{1}{2})
\end{equation}
after the normalization \(\vec{x}_3^{\dagger} G \vec{x}_3 = 1\), is a maximal state for \(s\in[0, \frac{1}{2})\). In other words, the state \(\ket{\Psi_{-}}\) in Eq.~\eqref{GU-3dim-positive}
can be transformed into any three-dimensional state with \(s\in[0, \frac{1}{2})\) via superposition-free operations.
It is possible to diversify the examples in these ways (see Table \ref{Table:D3GoldenUnits}).
For a given inner product value \(s\), the \(l_1\)-norm of superposition \cite{Plenio-RTofS} reaches its maximum value for the states given in Eq.~\eqref{GU-3dim-negative} when \(s\in(-\frac{1}{2},0]\) and Eq.\eqref{GU-3dim-positive} when \(s\in[0, \frac{1}{2})\) (see Fig.~\ref{Fig3:3D}).

So far, with the results we have presented above, we have seen that maximal states exist in RTS; however, there is not a maximal state for all inner product values in general~\cite{Plenio-RTofS}. As an example, consider the case \(s=1/2\), which was also discussed in Ref.~\cite{Plenio-RTofS}. The eigensystem of the Gram matrix is obtained such that
\begin{eqnarray}
\lambda_1 = 2, \quad
\lambda_2 = \lambda_3 = \frac{1}{2};
\end{eqnarray}
\begin{eqnarray}
\vec{x}_1 = {(1, 1, 1)}^{\intercal}, \quad \vec{x}_2 = {(1, 0, -1)}^{\intercal}, \quad \vec{x}_3 = {(1, -1, 0)}^{\intercal}. \ \
\end{eqnarray}
Since \(\lambda_{\text{min}} = \lambda_2 = \lambda_3 = 1/2\) for degenerate eigenvalues, any vector in the form \(a_1 \vec{x}_2 + a_2 \vec{x}_3 = {(a_1+a_2, -a_2, -a_1)}^{\intercal}\) is an eigenvector corresponding to the minimum eigenvalue. The candidate maximal state can be written as
\begin{equation}\label{eq:threelevel-maximal-OneOverTwo}
\ket{\psi} = \mathcal{N} \Big[(a_1+a_2)\ket{c_1} - a_2\ket{c_2} - a_1\ket{c_3}\Big],
\end{equation}
where \(\mathcal{N} = {(a_1^2 + a_1 a_2 + a_2^2)}^{-1/2}\). However, one can show that a state in this form cannot satisfy \(\tilde{\psi}_i = 1/d\) for all \(i=1, 2, 3\) for any \(a_1\) and \(a_2\). Therefore, there is no maximal state for \(s=1/2\), i.e., the inner product setting \(\braket{c_1}{c_2}=\braket{c_1}{c_3} =\braket{c_2}{c_3}=1/2\) violates Proposition \ref{Prop:1}. These explicit examples show that there are maximal states for \(d \geq 3\) in general; however, for a given inner product setting there may not be a maximal state as in the case of \(s=1/2\).

\begin{figure}[t]
	\centering
	\includegraphics[width=1\columnwidth]{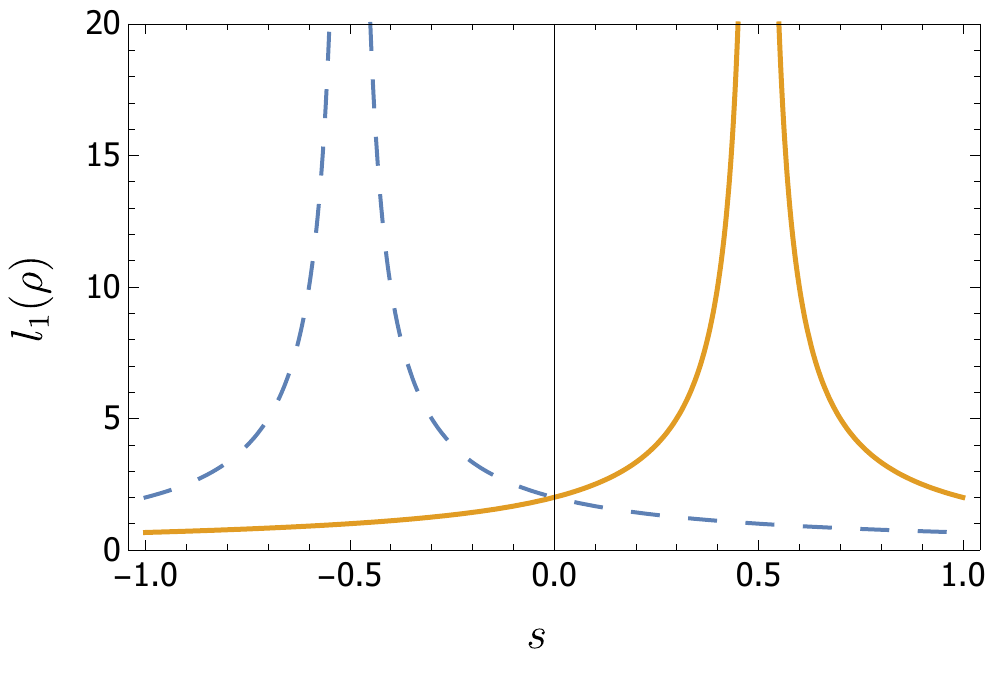}
	\caption{Plots of the \(l_1\) norm of superposition of the states \(\ket{\Psi_{3}^{+}}\) and \(\ket{\Psi_{3}^{-}}\) given in Eqs.~\eqref{GU-3dim-negative} and \eqref{GU-3dim-positive}. The \(l_1\)-norm of superposition of the state \(\ket{\Psi_{3}^{-}}\) is \(l_1(\rho_{\ket{\Psi_{3}^{-}}})={2}/{(1-2s)}\) (orange solid line) and the \(l_1\) -norm of superposition of the state \(\ket{\Psi_{3}^{+}}\) is \(l_1(\rho_{\ket{\Psi_{3}^{+}}})={2}/{(1+2s)}\) (blue dashed line). We also recall that the \(l_1\)-norm of coherence of the maximally coherent state \(\ket{\Psi_{3}}=(\ket{0}+\ket{1}+\ket{2})/\sqrt{3}\) is equal to  \({2}\) for \(d=3\), where \(l_1(\rho_{\ket{\Psi_{3}^{-}}})=l_1(\rho_{\ket{\Psi_{3}^{+}}})=2\) for \(s=0\). For any given state \(\ket{\phi}=\phi_1\ket{c_1}+\phi_2\ket{c_2}+\phi_3\ket{c_3}\) with inner products  \(\braket{c_1}{c_2}=\braket{c_1}{c_3}=\braket{c_2}{c_3}=s\) and state \(\ket{\chi}=\chi_1\ket{c_1}+\chi_2\ket{c_2}+\chi_3\ket{c_3}\) with inner products \(\braket{c_1}{c_2}=-s\) and \(\braket{c_1}{c_3}=\braket{c_2}{c_3}=s\) in dimension \(3\), we have the following results: \(l_1(\rho_{\ket{\Psi_{3}^{+}}}) > l_1(\rho_{\ket{\phi}})\) for \(s\in(-1/2, 0]\) and \(l_1(\rho_{\ket{\Psi_{3}^{-}}}) > l_1(\rho_{\ket{\chi}})\) for \(s\in[0, 1/2)\).}
	\label{Fig3:3D}
\end{figure}

Additionally, it is apparent that the Gram matrices considered above have a two-fold degenerate eigenvalue.
Although the reason for the degeneracy is due to the choice of the inner product settings,
it turns out that a Gram matrix must necessarily have degenerate eigenvalues to admit a maximal state.
Let us proceed to discuss the case \(d=3\). Since a Gram matrix is positive (semi)definite, it can be written as
\begin{equation}
  G = \sum_{i=1}^{3} \lambda_{i} \vec{x}_i\vec{x}_i^{\dagger},
\end{equation}
where \(\lambda_{i}\) are eigenvalues of \(G\) such that \(\lambda_{1} \leq \lambda_2 \leq \lambda_3\) with eigenvectors \(\vec{x_i}\). If one assumes that \(G\) admits a maximal state, then \(\vec{x}_1\) must be in the form of Eq.~\eqref{eq:maximalform-vec} up to a normalization constant. One then finds that the Gram matrix has the diagonal entries
\begin{equation}
  \begin{split}
    G_{11} = G_{22} &= \frac{1}{6}\left(2\lambda_1+3\lambda_2 + \lambda_3\right), \\
    G_{33} &= \frac{1}{3} \left(\lambda_1 + 2\lambda_3\right).
  \end{split}
\end{equation}
Since \(G_{11} = G_{22} = G_{33} = 1 \), one must then have \(\lambda_2 = \lambda_3 = (3 - \lambda_1)/2\), which implies degeneracy in the eigenvalues. Moreover, given that the eigenvalues of the Gram matrix can be expressed in terms of the minimum eigenvalue \(\lambda_1\), one can parametrize the set of Gram matrices that admit a maximal state, that is,
\begin{equation}
  \mathcal{G} = \left\{\sum_{i=1}^{3} \lambda_{i} \vec{x}_i\vec{x}_i^{\dagger} \,\, \bigg| \,\, \lambda_1 \in (0,1) \right\}.
\end{equation}

\subsection{\(d\)-dimensional systems}

As the dimension of systems increases, analysis of maximal states becomes slightly more complicated.
Here the main difficulty arises from the phase factors that are likely to emerge as a result of the inner products between basis states, where we have \(\braket{c_i}{c_j} = s_{ij} = \alpha_{ij}e^{i\beta_{ij}}\) (\(\alpha_{ij} \in \mathds{R}^{\geq 0} \) and \(\beta_{ij} \in [0, 2\pi)\)) in general. On the other hand, when we have inner product settings such that the inner products do not equal each other up to a phase factor, that is, \(\abs{s_{ij}} \neq s\) for all \(i \neq j\), then the given inner product setting does not admit a maximal state in general. Clearly, the case where the inner products between basis states have equal and real values, that is, \(\alpha_{ij}=s\) and \(\beta_{ij}=0\), is the most elementary one to figure out. Here we focus on this case and have the following proposition.

\begin{proposition}\label{Prop:3}
The state
\begin{equation}\label{GU:General-negative}
\ket{\Psi_{d}^{+}}:=\frac{1}{\sqrt{d \left(1 +(d-1)s\right)}} \sum_{j=1}^{d} \ket{c_j}, \quad s\in(\frac{1}{1-d}, 0],
\end{equation}
is a maximal state for an  equal and real inner product setting, that is, 
\(\braket{c_i}{c_j} = s, \, s \in \mathds{R}\), for \(i \neq j\).
\end{proposition}

Before we provide the proof of Proposition \ref{Prop:3}, we recall the following. Let \(A\) be a \(d \cross d\) matrix whose entries \(a_{ij} \in \mathds{C}\). If \(\abs{a_{ii}} \geq \sum_{i \neq j} \abs{a_{ij}}\) for all \(i=1, 2, \dots, d\), then \(A\) is called a diagonally dominant matrix. If \(A\) is also Hermitian with real non-negative diagonal entries , then the matrix \(A\) is positive semidefinite  \cite{Horn-GramMatrix,VargaMatrixIterative2000}.

\begin{proof}
We start the proof by showing that the superposition-free Kraus operators in the set \(S_1\) defined in Eq.~\eqref{eq:dlevel-rankd-kraus} and a set of additional superposition-free Kraus operators constitute a trace-preserving operation. The equation \(G - \sum_{n=1}^{d!}\tilde{K}_n^{\dagger} G \tilde{K}_n \geq 0\)
must be positive semidefinite so that one can decompose it into positive semidefinite operators in the form of \(\sum_{m=1}^{d}\tilde{F}_m^{\dagger} G \tilde{F}_m\) and hence \(\sum_{n=1}^{d!}\tilde{K}_n^{\dagger} G \tilde{K}_n+\sum_{m=1}^{d}\tilde{F}_m^{\dagger} G \tilde{F}_m=G\). Since \(\tilde{F}_m\vec{\psi} = 0\) for all \(m\), one also requires
\begin{equation}\label{dlevel-zero-eigenvalue}
    \Big(G - \sum_{n=1}^{d!}\tilde{K}_n^{\dagger} G \tilde{K}_n\Big)\vec{\psi} = 0.
\end{equation}
Now assume that \(G \vec{\psi} = \lambda_{\text{min}}\vec{\psi}\) and \(\tilde{\psi}={(1/d, \dots, 1/d)}^{\intercal}\).
To see that Eq.~\eqref{dlevel-zero-eigenvalue} holds, consider the following:
\begin{equation}\label{G-Y-firsteq}
    \begin{split}
    {\Big((G-\mathds{K})\vec{\psi}\Big)}_{i} &= \sum_{l=1}^{d} G_{il} e^{i\theta_l}\psi_l - \bigg[\frac{1}{d}\frac{\phi^2}{\psi_i^2}\sum_{l=1}^{d} e^{i\theta_l}\psi_l\delta_{il} \\ &+ \frac{(d-2)!}{d!}(1 - \phi^{2})\sum_{l=1}^{d} \frac{e^{i\theta_{i-l}}}{\psi_i\psi_l} e^{i\theta_l}\psi_l (1 - \delta_{il}) \bigg],
    \end{split}
\end{equation}
where \( \phi^2 \coloneqq \sum_{j=1}^{d}\phi_j^2\) and the matrix \(\mathds{K}\) is given by Eqs.~\eqref{KnSecondSet11} and \eqref{KnSecondSet12}. From Eq.~\eqref{eq:maximalform-vec} we have that \({\psi_i}^2 = {1}/{(d \lambda_{\min})}\) and then Eq.~\eqref{G-Y-firsteq} becomes
\begin{equation}
    \begin{split}
    {\Big((G-\mathds{K})\vec{\psi}\Big)}_{i} &= \sum_{l=1}^{d} G_{il} e^{i\theta_l}\psi_l - \lambda_{\text{min}}e^{i\theta_i}\psi_i.
    \end{split}
\end{equation}
From \(G \vec{\psi} = \lambda_{\text{min}}\vec{\psi}\) we have \(\sum_{l=1}^{d} G_{il} e^{i\theta_l}\psi_l = \lambda_{\text{min}}e^{i\theta_i}\psi_i\) and hence \((G-\mathds{K})\vec{\psi} = 0\), which proves Eq.~\eqref{dlevel-zero-eigenvalue}.

We now show that \(G-\mathds{K}\) is positive semidefinite. By combining Eq.~\eqref{KnSecondSet11} and \({\psi_i}^2 = {1}/{(d \lambda_{\min})}\), the diagonal entries of the matrix \(G-Y\) are written such that
\begin{equation}\label{eq:dlevelX-diagonal}
    \big(G-\mathds{K}\big)_{ii} = 1 - \lambda_{\text{min}} \phi^2,
\end{equation}
since \(G_{ii} = 1\) for all \(i\). Looking at Eq.~\eqref{eq:quadratic-identity}, it is easy to see that Eq.~\eqref{eq:dlevelX-diagonal} is non-negative. Also, the matrix \(G-\mathds{K}\) has the following off-diagonal entries:
\begin{equation}\label{eq:dlevelX-offdiagonal}
    \big(G-\mathds{K}\big)_{il} = G_{il} - \frac{\lambda_{\text{min}}}{d-1}\left(1-\phi^2\right)e^{i\theta_{i-l}}.
\end{equation}

First, we consider qubit systems with \(\braket{c_1}{c_2} = se^{i\theta}\), where the minimum eigenvalue of
\(G\) is \(1-s\) when \(s \in [0,1)\) and \(1+s\) when \(s \in (-1,0]\). In the case where \(\lambda_{\text{min}} = 1-s\), Eqs.~\eqref{eq:dlevelX-diagonal} and \eqref{eq:dlevelX-offdiagonal} provide \(\big(G-\mathds{K}\big)_{11} = \big(G-\mathds{K}\big)_{22} = 1 - (1-s)\phi^2\) and \(\big(G-\mathds{K}\big)_{12} = \big(G-\mathds{K}\big)_{21}^{*} = \big(1-(1-s)\phi^2\big)e^{i\theta}\), respectively. It is obvious that \(G-\mathds{K}\) is a diagonally dominant matrix. Since \(G-\mathds{K}\) is also Hermitian with positive diagonal entries, one concludes that it is positive semidefinite. It is straightforward to show that the resultant matrix is positive semidefinite also for the case \(\lambda_{\text{min}} = 1+s\).

The same approach can be used to prove that the states given in Table \ref{Table:D3GoldenUnits} are maximal. Here, for instance, we provide a proof for the case \(\{s_{12}, s_{13}, s_{23}\}=\{s, -is, is\}\) with \(s\in [0,\frac{1}{2})\). In this case, we have \(\big(G-\mathds{K}\big)_{ii} = 1 - (1-2s)\phi^2\) for \(i=1,2,3\) and \(\abs{\big(G-\mathds{K}\big)_{il}} =
{(1-(1-2s)\phi^2)/2}\) for \(i,l=1,2,3\) with \(i\neq l\). By using the same reasoning as above for the qubit case (diagonally dominant matrix and Hermitian matrix), one concludes that \(G-\mathds{K}\) is a positive-semidefinite matrix.

Second, we consider \(d\geq 3\) level systems with a real and equal inner product setting. The Gram matrix then can be written as
\begin{equation}\label{eq:grammatrix-realcase-eigs}
    G = (1-s)\mathds{1} + s \vec{\psi}\vec{\psi}^{\intercal},
\end{equation}
where \(\mathds{1}\) is a \(d \times d\) identity matrix and \(\vec{\psi}\) is a column vector whose elements are all one. It is clear that \(\vec{\psi}\) is an eigenvector corresponding to the eigenvalue \(1 + s(d-1)\) with multiplicity one and any vector orthogonal to \(\vec{\psi}\) is an eigenvector corresponding to the eigenvalue \(1-s\) with multiplicity \(d-1\). Then we have that \(\lambda_{\text{min}} = 1 + s(d-1)\) for the interval \(1/(1-d) < s \leq 0\) (in accordance with the assumption \(G \vec{\psi} = \lambda_{\text{min}}\vec{\psi}\)). By combining these with Eqs.~\eqref{eq:dlevelX-diagonal} and \eqref{eq:dlevelX-offdiagonal}, we obtain
\begin{eqnarray}\label{eq:dlevelX-diagonal-real}
\begin{split}
\big(G-\mathds{K}\big)_{ii} &= 1 - \phi^2\big(1+s(d-1)\big),\\
\big(G-\mathds{K}\big)_{il} &= \frac{\big(G-\mathds{K}\big)_{ii}}{1-d}.
\end{split}
\end{eqnarray}
Note that \(\abs{\big(G-\mathds{K}\big)_{ii}} = \sum_{l} \abs{\big(G-\mathds{K}\big)_{il}} \) for all \(l \neq i\) and \(i=1, 2, \dots, d\). Hence, using the dominant diagonal property of the matrix, we can conclude that it is positive semidefinite.

Moreover, an arbitrary mixed superposition state can be written as \(\sigma = \sum_i p_i \ket{\phi_i}\bra{\phi_i}\),
where \(p_i \geq 0\) forms a probability distribution. Then, for an inner product setting supported by Proposition \ref{Prop:2} and Proposition \ref{Prop:3}, one can implement the following superposition-free operation:
\(\Phi_i(\ketbra{\psi}) = \ket{\phi_i}\bra{\phi_i}\), where \(\ket{\psi}\) is the maximal state of the corresponding inner
product setting. Then, it is clear that \(\sum_i p_i \Phi_i(\ketbra{\psi}) = \sigma\). This completes the proof, since one does not need to built operators in the second set \(S_2\) explicitly (Theorem 7 of Ref.~\cite{Plenio-RTofS}). Although that is the case, the entries of the operators in the second set can be determined by following the steps described in the Supplemental Material of Ref.~\cite{Plenio-RTofS}.
\end{proof}

Finally, the constant trace condition introduced in Ref.~\cite{Liu-OneShotRT} can be easily generalized for dimension \(d\).
For \(d\)-dimensional systems with \(d\) pure basis states \(\{\ket{c_1}, \ket{c_2}, \dots, \ket{c_d}\}\), one can obtain
\(\text{tr}({\hat{\Psi}_{d}^{+}}\ket{c_i}\bra{c_i}) = \lambda_{\min}/d\) for \(i=1,2,\dots,d\). Here we write \({\hat{\Psi}_{d}^{+}}\) as a shorthand for \(\ket{\Psi_{d}^{+}}\bra{\Psi_{d}^{+}}\).
Note that for the state given in Eq.~\eqref{GU:General-negative}, \(\lambda_{\min}=1+(d-1)s\) for \(s \in (\frac{1}{1-d},0]\). This ensures \(\text{tr}({\hat{\Psi}_{d}^{+}}\delta) = \text{const}\), \(\forall\delta \in \mathcal{F}\) due to the linearity of the trace function.

\section{Conclusions}\label{Sec:Conc}

In this work we developed a structural way to investigate maximal states of the resource theory of superposition \cite{Plenio-RTofS}, which is a generalization of the resource theory of coherence \cite{Baumgratz-Coherence}. For this purpose, we used the Gram matrix to represent different sets of linearly independent (nonorthogonal) basis states. We provided a necessary condition for the existence of a maximal state for a set of arbitrary basis states utilizing the eigenvalues and eigenvectors of the corresponding Gram matrix. An immediate corollary of this condition provides the form of (candidate) maximal states which reduces to the form of the maximal state of the resource theory of coherence in the orthonormal limit, i.e., \(G\) (Gram matrix) \(\rightarrow\) \(\mathds{1}\) (identity).

In light of the proposed criterion, we first investigated the maximal states of qubit systems. Our findings show that every possible linearly independent basis states of a qubit system admits a maximal state, which unifies the results of Refs.~\cite{Plenio-RTofS,torun2020resource}. Then, we discussed three-dimensional systems in detail along with \(d\)-dimensional systems.
For \(d \geq 3\), the existence of a maximal state depends on the inner products (overlaps) of basis states; hence, only a subset of all possible basis states admits a maximal state. To have seen that the Gram matrix has such an effective feature shows that it may have an indispensable role in other kinds of connections between coherence \cite{Baumgratz-Coherence} and superposition \cite{Plenio-RTofS}.

Future studies may benefit from our work in various contexts. For instance, a further generalization of coherence distillation \cite{2018Regula-OneShotCohDis,2019Torunn-CohDis,2019Liu-DetCoh,2020RegulaOneShotDis,Pang2020} to the superposition of linearly independent nonorthogonal basis states would be a worthwhile direction for future investigation, which is particularly important for quantum cryptographic protocols. It would be interesting to explore whether the Gram matrix has a convenient function to adapt such works to the case of superposition of nonorthogonal basis states.

\begin{acknowledgments}
We are grateful to Ali Yildiz and Bartosz Regula for helpful discussions and comments. G.T. acknowledges support from the Scientific and Technological Research Council of Turkey (TUBITAK) (Grant No. 120F089).
\end{acknowledgments}


\begin{thebibliography}{37}%
\makeatletter
\providecommand \@ifxundefined [1]{%
 \@ifx{#1\undefined}
}%
\providecommand \@ifnum [1]{%
 \ifnum #1\expandafter \@firstoftwo
 \else \expandafter \@secondoftwo
 \fi
}%
\providecommand \@ifx [1]{%
 \ifx #1\expandafter \@firstoftwo
 \else \expandafter \@secondoftwo
 \fi
}%
\providecommand \natexlab [1]{#1}%
\providecommand \enquote  [1]{``#1''}%
\providecommand \bibnamefont  [1]{#1}%
\providecommand \bibfnamefont [1]{#1}%
\providecommand \citenamefont [1]{#1}%
\providecommand \href@noop [0]{\@secondoftwo}%
\providecommand \href [0]{\begingroup \@sanitize@url \@href}%
\providecommand \@href[1]{\@@startlink{#1}\@@href}%
\providecommand \@@href[1]{\endgroup#1\@@endlink}%
\providecommand \@sanitize@url [0]{\catcode `\\12\catcode `\$12\catcode
  `\&12\catcode `\#12\catcode `\^12\catcode `\_12\catcode `\%12\relax}%
\providecommand \@@startlink[1]{}%
\providecommand \@@endlink[0]{}%
\providecommand \url  [0]{\begingroup\@sanitize@url \@url }%
\providecommand \@url [1]{\endgroup\@href {#1}{\urlprefix }}%
\providecommand \urlprefix  [0]{URL }%
\providecommand \Eprint [0]{\href }%
\providecommand \doibase [0]{https://doi.org/}%
\providecommand \selectlanguage [0]{\@gobble}%
\providecommand \bibinfo  [0]{\@secondoftwo}%
\providecommand \bibfield  [0]{\@secondoftwo}%
\providecommand \translation [1]{[#1]}%
\providecommand \BibitemOpen [0]{}%
\providecommand \bibitemStop [0]{}%
\providecommand \bibitemNoStop [0]{.\EOS\space}%
\providecommand \EOS [0]{\spacefactor3000\relax}%
\providecommand \BibitemShut  [1]{\csname bibitem#1\endcsname}%
\let\auto@bib@innerbib\@empty
\bibitem [{\citenamefont {Dirac}(1930)}]{Dirac-Superposition}%
  \BibitemOpen
  \bibfield  {author} {\bibinfo {author} {\bibfnamefont {P.~A.~M.}\
  \bibnamefont {Dirac}},\ }\href@noop {} {\emph {\bibinfo {title} {The
  Principles of Quantum Mechanics}}},\ \bibinfo {edition} {3rd}\ ed.\ (\bibinfo
   {publisher} {Clarendon Press, Oxford},\ \bibinfo {year} {1930})\BibitemShut
  {NoStop}%
\bibitem [{\citenamefont {Horodecki}\ \emph {et~al.}(2009)\citenamefont
  {Horodecki}, \citenamefont {Horodecki}, \citenamefont {Horodecki},\ and\
  \citenamefont {Horodecki}}]{Horodecki-QE}%
  \BibitemOpen
  \bibfield  {author} {\bibinfo {author} {\bibfnamefont {R.}~\bibnamefont
  {Horodecki}}, \bibinfo {author} {\bibfnamefont {P.}~\bibnamefont
  {Horodecki}}, \bibinfo {author} {\bibfnamefont {M.}~\bibnamefont
  {Horodecki}},\ and\ \bibinfo {author} {\bibfnamefont {K.}~\bibnamefont
  {Horodecki}},\ }\bibfield  {title} {\bibinfo {title} {Quantum entanglement},\
  }\href {https://doi.org/10.1103/RevModPhys.81.865} {\bibfield  {journal}
  {\bibinfo  {journal} {Rev. Mod. Phys.}\ }\textbf {\bibinfo {volume} {81}},\
  \bibinfo {pages} {865} (\bibinfo {year} {2009})}\BibitemShut {NoStop}%
\bibitem [{\citenamefont {Baumgratz}\ \emph {et~al.}(2014)\citenamefont
  {Baumgratz}, \citenamefont {Cramer},\ and\ \citenamefont
  {Plenio}}]{Baumgratz-Coherence}%
  \BibitemOpen
  \bibfield  {author} {\bibinfo {author} {\bibfnamefont {T.}~\bibnamefont
  {Baumgratz}}, \bibinfo {author} {\bibfnamefont {M.}~\bibnamefont {Cramer}},\
  and\ \bibinfo {author} {\bibfnamefont {M.~B.}\ \bibnamefont {Plenio}},\
  }\bibfield  {title} {\bibinfo {title} {Quantifying coherence},\ }\href
  {https://doi.org/10.1103/PhysRevLett.113.140401} {\bibfield  {journal}
  {\bibinfo  {journal} {Phys. Rev. Lett.}\ }\textbf {\bibinfo {volume} {113}},\
  \bibinfo {pages} {140401} (\bibinfo {year} {2014})}\BibitemShut {NoStop}%
\bibitem [{\citenamefont {Hickey}\ and\ \citenamefont
  {Gour}(2018)}]{Hickey_2018}%
  \BibitemOpen
  \bibfield  {author} {\bibinfo {author} {\bibfnamefont {A.}~\bibnamefont
  {Hickey}}\ and\ \bibinfo {author} {\bibfnamefont {G.}~\bibnamefont {Gour}},\
  }\bibfield  {title} {\bibinfo {title} {Quantifying the imaginarity of quantum
  mechanics},\ }\href {https://doi.org/10.1088/1751-8121/aabe9c} {\bibfield
  {journal} {\bibinfo  {journal} {J. Phys. A: Math. Theor.}\ }\textbf {\bibinfo
  {volume} {51}},\ \bibinfo {pages} {414009} (\bibinfo {year}
  {2018})}\BibitemShut {NoStop}%
\bibitem [{\citenamefont {Brunner}\ \emph {et~al.}(2014)\citenamefont
  {Brunner}, \citenamefont {Cavalcanti}, \citenamefont {Pironio}, \citenamefont
  {Scarani},\ and\ \citenamefont {Wehner}}]{2014Bellnonlocality}%
  \BibitemOpen
  \bibfield  {author} {\bibinfo {author} {\bibfnamefont {N.}~\bibnamefont
  {Brunner}}, \bibinfo {author} {\bibfnamefont {D.}~\bibnamefont {Cavalcanti}},
  \bibinfo {author} {\bibfnamefont {S.}~\bibnamefont {Pironio}}, \bibinfo
  {author} {\bibfnamefont {V.}~\bibnamefont {Scarani}},\ and\ \bibinfo {author}
  {\bibfnamefont {S.}~\bibnamefont {Wehner}},\ }\bibfield  {title} {\bibinfo
  {title} {Bell nonlocality},\ }\href
  {https://doi.org/10.1103/RevModPhys.86.419} {\bibfield  {journal} {\bibinfo
  {journal} {Rev. Mod. Phys.}\ }\textbf {\bibinfo {volume} {86}},\ \bibinfo
  {pages} {419} (\bibinfo {year} {2014})}\BibitemShut {NoStop}%
\bibitem [{\citenamefont {Winter}\ and\ \citenamefont
  {Yang}(2016)}]{Winter120404}%
  \BibitemOpen
  \bibfield  {author} {\bibinfo {author} {\bibfnamefont {A.}~\bibnamefont
  {Winter}}\ and\ \bibinfo {author} {\bibfnamefont {D.}~\bibnamefont {Yang}},\
  }\bibfield  {title} {\bibinfo {title} {Operational resource theory of
  coherence},\ }\href {https://doi.org/10.1103/PhysRevLett.116.120404}
  {\bibfield  {journal} {\bibinfo  {journal} {Phys. Rev. Lett.}\ }\textbf
  {\bibinfo {volume} {116}},\ \bibinfo {pages} {120404} (\bibinfo {year}
  {2016})}\BibitemShut {NoStop}%
\bibitem [{\citenamefont {Takagi}\ and\ \citenamefont
  {Regula}(2019)}]{Regula-QRTs}%
  \BibitemOpen
  \bibfield  {author} {\bibinfo {author} {\bibfnamefont {R.}~\bibnamefont
  {Takagi}}\ and\ \bibinfo {author} {\bibfnamefont {B.}~\bibnamefont
  {Regula}},\ }\bibfield  {title} {\bibinfo {title} {General resource theories
  in quantum mechanics and beyond: Operational characterization via
  discrimination tasks},\ }\href {https://doi.org/10.1103/PhysRevX.9.031053}
  {\bibfield  {journal} {\bibinfo  {journal} {Phys. Rev. X}\ }\textbf {\bibinfo
  {volume} {9}},\ \bibinfo {pages} {031053} (\bibinfo {year}
  {2019})}\BibitemShut {NoStop}%
\bibitem [{\citenamefont {Oszmaniec}\ and\ \citenamefont
  {Biswas}(2019)}]{Oszmaniec2019operational}%
  \BibitemOpen
  \bibfield  {author} {\bibinfo {author} {\bibfnamefont {M.}~\bibnamefont
  {Oszmaniec}}\ and\ \bibinfo {author} {\bibfnamefont {T.}~\bibnamefont
  {Biswas}},\ }\bibfield  {title} {\bibinfo {title} {Operational relevance of
  resource theories of quantum measurements},\ }\href
  {https://doi.org/10.22331/q-2019-04-26-133} {\bibfield  {journal} {\bibinfo
  {journal} {{Quantum}}\ }\textbf {\bibinfo {volume} {3}},\ \bibinfo {pages}
  {133} (\bibinfo {year} {2019})}\BibitemShut {NoStop}%
\bibitem [{\citenamefont {Liu}\ \emph {et~al.}(2019)\citenamefont {Liu},
  \citenamefont {Bu},\ and\ \citenamefont {Takagi}}]{Liu-OneShotRT}%
  \BibitemOpen
  \bibfield  {author} {\bibinfo {author} {\bibfnamefont {Z.-W.}\ \bibnamefont
  {Liu}}, \bibinfo {author} {\bibfnamefont {K.}~\bibnamefont {Bu}},\ and\
  \bibinfo {author} {\bibfnamefont {R.}~\bibnamefont {Takagi}},\ }\bibfield
  {title} {\bibinfo {title} {One-shot operational quantum resource theory},\
  }\href {https://doi.org/10.1103/PhysRevLett.123.020401} {\bibfield  {journal}
  {\bibinfo  {journal} {Phys. Rev. Lett.}\ }\textbf {\bibinfo {volume} {123}},\
  \bibinfo {pages} {020401} (\bibinfo {year} {2019})}\BibitemShut {NoStop}%
\bibitem [{\citenamefont {Takagi}\ \emph {et~al.}(2019)\citenamefont {Takagi},
  \citenamefont {Regula}, \citenamefont {Bu}, \citenamefont {Liu},\ and\
  \citenamefont {Adesso}}]{Takagi_OpAdvQR}%
  \BibitemOpen
  \bibfield  {author} {\bibinfo {author} {\bibfnamefont {R.}~\bibnamefont
  {Takagi}}, \bibinfo {author} {\bibfnamefont {B.}~\bibnamefont {Regula}},
  \bibinfo {author} {\bibfnamefont {K.}~\bibnamefont {Bu}}, \bibinfo {author}
  {\bibfnamefont {Z.-W.}\ \bibnamefont {Liu}},\ and\ \bibinfo {author}
  {\bibfnamefont {G.}~\bibnamefont {Adesso}},\ }\bibfield  {title} {\bibinfo
  {title} {Operational advantage of quantum resources in subchannel
  discrimination},\ }\href {https://doi.org/10.1103/PhysRevLett.122.140402}
  {\bibfield  {journal} {\bibinfo  {journal} {Phys. Rev. Lett.}\ }\textbf
  {\bibinfo {volume} {122}},\ \bibinfo {pages} {140402} (\bibinfo {year}
  {2019})}\BibitemShut {NoStop}%
\bibitem [{\citenamefont {Liu}\ and\ \citenamefont {Yuan}(2020)}]{RT-of-QC}%
  \BibitemOpen
  \bibfield  {author} {\bibinfo {author} {\bibfnamefont {Y.}~\bibnamefont
  {Liu}}\ and\ \bibinfo {author} {\bibfnamefont {X.}~\bibnamefont {Yuan}},\
  }\bibfield  {title} {\bibinfo {title} {Operational resource theory of quantum
  channels},\ }\href {https://doi.org/10.1103/PhysRevResearch.2.012035}
  {\bibfield  {journal} {\bibinfo  {journal} {Phys. Rev. Research}\ }\textbf
  {\bibinfo {volume} {2}},\ \bibinfo {pages} {012035} (\bibinfo {year}
  {2020})}\BibitemShut {NoStop}%
\bibitem [{\citenamefont {Li}\ \emph {et~al.}(2020)\citenamefont {Li},
  \citenamefont {Bu},\ and\ \citenamefont {Liu}}]{Li-QuantifyingRCQC}%
  \BibitemOpen
  \bibfield  {author} {\bibinfo {author} {\bibfnamefont {L.}~\bibnamefont
  {Li}}, \bibinfo {author} {\bibfnamefont {K.}~\bibnamefont {Bu}},\ and\
  \bibinfo {author} {\bibfnamefont {Z.-W.}\ \bibnamefont {Liu}},\ }\bibfield
  {title} {\bibinfo {title} {Quantifying the resource content of quantum
  channels: An operational approach},\ }\href
  {https://doi.org/10.1103/PhysRevA.101.022335} {\bibfield  {journal} {\bibinfo
   {journal} {Phys. Rev. A}\ }\textbf {\bibinfo {volume} {101}},\ \bibinfo
  {pages} {022335} (\bibinfo {year} {2020})}\BibitemShut {NoStop}%
\bibitem [{\citenamefont {Wu}\ \emph {et~al.}(2021)\citenamefont {Wu},
  \citenamefont {Kondra}, \citenamefont {Rana}, \citenamefont {Scandolo},
  \citenamefont {Xiang}, \citenamefont {Li}, \citenamefont {Guo},\ and\
  \citenamefont {Streltsov}}]{Wu2021-Imaginarity}%
  \BibitemOpen
  \bibfield  {author} {\bibinfo {author} {\bibfnamefont {K.-D.}\ \bibnamefont
  {Wu}}, \bibinfo {author} {\bibfnamefont {T.~V.}\ \bibnamefont {Kondra}},
  \bibinfo {author} {\bibfnamefont {S.}~\bibnamefont {Rana}}, \bibinfo {author}
  {\bibfnamefont {C.~M.}\ \bibnamefont {Scandolo}}, \bibinfo {author}
  {\bibfnamefont {G.-Y.}\ \bibnamefont {Xiang}}, \bibinfo {author}
  {\bibfnamefont {C.-F.}\ \bibnamefont {Li}}, \bibinfo {author} {\bibfnamefont
  {G.-C.}\ \bibnamefont {Guo}},\ and\ \bibinfo {author} {\bibfnamefont
  {A.}~\bibnamefont {Streltsov}},\ }\bibfield  {title} {\bibinfo {title}
  {Operational resource theory of imaginarity},\ }\href
  {https://doi.org/10.1103/PhysRevLett.126.090401} {\bibfield  {journal}
  {\bibinfo  {journal} {Phys. Rev. Lett.}\ }\textbf {\bibinfo {volume} {126}},\
  \bibinfo {pages} {090401} (\bibinfo {year} {2021})}\BibitemShut {NoStop}%
\bibitem [{\citenamefont {Kuroiwa}\ and\ \citenamefont
  {Yamasaki}(2020)}]{Kuroiwa2020generalquantum}%
  \BibitemOpen
  \bibfield  {author} {\bibinfo {author} {\bibfnamefont {K.}~\bibnamefont
  {Kuroiwa}}\ and\ \bibinfo {author} {\bibfnamefont {H.}~\bibnamefont
  {Yamasaki}},\ }\bibfield  {title} {\bibinfo {title} {General {Q}uantum
  {R}esource {T}heories: {D}istillation, {F}ormation and {C}onsistent
  {R}esource {M}easures},\ }\href {https://doi.org/10.22331/q-2020-11-01-355}
  {\bibfield  {journal} {\bibinfo  {journal} {{Quantum}}\ }\textbf {\bibinfo
  {volume} {4}},\ \bibinfo {pages} {355} (\bibinfo {year} {2020})}\BibitemShut
  {NoStop}%
\bibitem [{\citenamefont {Zhou}\ and\ \citenamefont
  {Buscemi}(2020)}]{Zhou2020SC}%
  \BibitemOpen
  \bibfield  {author} {\bibinfo {author} {\bibfnamefont {W.}~\bibnamefont
  {Zhou}}\ and\ \bibinfo {author} {\bibfnamefont {F.}~\bibnamefont {Buscemi}},\
  }\bibfield  {title} {\bibinfo {title} {General state transitions with exact
  resource morphisms: a unified resource-theoretic approach},\ }\href
  {https://doi.org/10.1088/1751-8121/abafe5} {\bibfield  {journal} {\bibinfo
  {journal} {J. Phys. A: Math. Theor.}\ }\textbf {\bibinfo {volume} {53}},\
  \bibinfo {pages} {445303} (\bibinfo {year} {2020})}\BibitemShut {NoStop}%
\bibitem [{\citenamefont {Contreras-Tejada}\ \emph {et~al.}(2019)\citenamefont
  {Contreras-Tejada}, \citenamefont {Palazuelos},\ and\ \citenamefont
  {de~Vicente}}]{Tejada2019MRE}%
  \BibitemOpen
  \bibfield  {author} {\bibinfo {author} {\bibfnamefont {P.}~\bibnamefont
  {Contreras-Tejada}}, \bibinfo {author} {\bibfnamefont {C.}~\bibnamefont
  {Palazuelos}},\ and\ \bibinfo {author} {\bibfnamefont {J.~I.}\ \bibnamefont
  {de~Vicente}},\ }\bibfield  {title} {\bibinfo {title} {Resource theory of
  entanglement with a unique multipartite maximally entangled state},\ }\href
  {https://doi.org/10.1103/PhysRevLett.122.120503} {\bibfield  {journal}
  {\bibinfo  {journal} {Phys. Rev. Lett.}\ }\textbf {\bibinfo {volume} {122}},\
  \bibinfo {pages} {120503} (\bibinfo {year} {2019})}\BibitemShut {NoStop}%
\bibitem [{\citenamefont {Peng}\ \emph {et~al.}(2016)\citenamefont {Peng},
  \citenamefont {Jiang},\ and\ \citenamefont {Fan}}]{Peng2016MCS}%
  \BibitemOpen
  \bibfield  {author} {\bibinfo {author} {\bibfnamefont {Y.}~\bibnamefont
  {Peng}}, \bibinfo {author} {\bibfnamefont {Y.}~\bibnamefont {Jiang}},\ and\
  \bibinfo {author} {\bibfnamefont {H.}~\bibnamefont {Fan}},\ }\bibfield
  {title} {\bibinfo {title} {Maximally coherent states and coherence-preserving
  operations},\ }\href {https://doi.org/10.1103/PhysRevA.93.032326} {\bibfield
  {journal} {\bibinfo  {journal} {Phys. Rev. A}\ }\textbf {\bibinfo {volume}
  {93}},\ \bibinfo {pages} {032326} (\bibinfo {year} {2016})}\BibitemShut
  {NoStop}%
\bibitem [{\citenamefont {Aberg}(2006)}]{Aberg-Superposition}%
  \BibitemOpen
  \bibfield  {author} {\bibinfo {author} {\bibfnamefont {J.}~\bibnamefont
  {Aberg}},\ }\href@noop {} {\bibinfo {title} {Quantifying superposition}}
  (\bibinfo {year} {2006}),\ \Eprint {https://arxiv.org/abs/quant-ph/0612146}
  {arXiv:quant-ph/0612146} \BibitemShut {NoStop}%
\bibitem [{\citenamefont {Theurer}\ \emph {et~al.}(2017)\citenamefont
  {Theurer}, \citenamefont {Killoran}, \citenamefont {Egloff},\ and\
  \citenamefont {Plenio}}]{Plenio-RTofS}%
  \BibitemOpen
  \bibfield  {author} {\bibinfo {author} {\bibfnamefont {T.}~\bibnamefont
  {Theurer}}, \bibinfo {author} {\bibfnamefont {N.}~\bibnamefont {Killoran}},
  \bibinfo {author} {\bibfnamefont {D.}~\bibnamefont {Egloff}},\ and\ \bibinfo
  {author} {\bibfnamefont {M.~B.}\ \bibnamefont {Plenio}},\ }\bibfield  {title}
  {\bibinfo {title} {Resource theory of superposition},\ }\href
  {https://doi.org/10.1103/PhysRevLett.119.230401} {\bibfield  {journal}
  {\bibinfo  {journal} {Phys. Rev. Lett.}\ }\textbf {\bibinfo {volume} {119}},\
  \bibinfo {pages} {230401} (\bibinfo {year} {2017})}\BibitemShut {NoStop}%
\bibitem [{\citenamefont {Torun}\ \emph {et~al.}(2021)\citenamefont {Torun},
  \citenamefont {\ifmmode \mbox{\c{S}}\else
  \c{S}\fi{}enya\ifmmode~\mbox{\c{s}}\else \c{s}\fi{}a},\ and\ \citenamefont
  {Yildiz}}]{torun2020resource}%
  \BibitemOpen
  \bibfield  {author} {\bibinfo {author} {\bibfnamefont {G.}~\bibnamefont
  {Torun}}, \bibinfo {author} {\bibfnamefont {H.~T.}\ \bibnamefont {\ifmmode
  \mbox{\c{S}}\else \c{S}\fi{}enya\ifmmode~\mbox{\c{s}}\else \c{s}\fi{}a}},\
  and\ \bibinfo {author} {\bibfnamefont {A.}~\bibnamefont {Yildiz}},\
  }\bibfield  {title} {\bibinfo {title} {Resource theory of superposition:
  State transformations},\ }\href {https://doi.org/10.1103/PhysRevA.103.032416}
  {\bibfield  {journal} {\bibinfo  {journal} {Phys. Rev. A}\ }\textbf {\bibinfo
  {volume} {103}},\ \bibinfo {pages} {032416} (\bibinfo {year}
  {2021})}\BibitemShut {NoStop}%
\bibitem [{\citenamefont {Bennett}\ \emph {et~al.}(1993)\citenamefont
  {Bennett}, \citenamefont {Brassard}, \citenamefont {Cr\'epeau}, \citenamefont
  {Jozsa}, \citenamefont {Peres},\ and\ \citenamefont
  {Wootters}}]{Bennett1993}%
  \BibitemOpen
  \bibfield  {author} {\bibinfo {author} {\bibfnamefont {C.~H.}\ \bibnamefont
  {Bennett}}, \bibinfo {author} {\bibfnamefont {G.}~\bibnamefont {Brassard}},
  \bibinfo {author} {\bibfnamefont {C.}~\bibnamefont {Cr\'epeau}}, \bibinfo
  {author} {\bibfnamefont {R.}~\bibnamefont {Jozsa}}, \bibinfo {author}
  {\bibfnamefont {A.}~\bibnamefont {Peres}},\ and\ \bibinfo {author}
  {\bibfnamefont {W.~K.}\ \bibnamefont {Wootters}},\ }\bibfield  {title}
  {\bibinfo {title} {Teleporting an unknown quantum state via dual classical
  and einstein-podolsky-rosen channels},\ }\href
  {https://doi.org/10.1103/PhysRevLett.70.1895} {\bibfield  {journal} {\bibinfo
   {journal} {Phys. Rev. Lett.}\ }\textbf {\bibinfo {volume} {70}},\ \bibinfo
  {pages} {1895} (\bibinfo {year} {1993})}\BibitemShut {NoStop}%
\bibitem [{\citenamefont {Luo}\ \emph {et~al.}(2019)\citenamefont {Luo},
  \citenamefont {Zhong}, \citenamefont {Erhard}, \citenamefont {Wang},
  \citenamefont {Peng}, \citenamefont {Krenn}, \citenamefont {Jiang},
  \citenamefont {Li}, \citenamefont {Liu}, \citenamefont {Lu}, \citenamefont
  {Zeilinger},\ and\ \citenamefont {Pan}}]{2019QTeleportationLuo}%
  \BibitemOpen
  \bibfield  {author} {\bibinfo {author} {\bibfnamefont {Y.-H.}\ \bibnamefont
  {Luo}}, \bibinfo {author} {\bibfnamefont {H.-S.}\ \bibnamefont {Zhong}},
  \bibinfo {author} {\bibfnamefont {M.}~\bibnamefont {Erhard}}, \bibinfo
  {author} {\bibfnamefont {X.-L.}\ \bibnamefont {Wang}}, \bibinfo {author}
  {\bibfnamefont {L.-C.}\ \bibnamefont {Peng}}, \bibinfo {author}
  {\bibfnamefont {M.}~\bibnamefont {Krenn}}, \bibinfo {author} {\bibfnamefont
  {X.}~\bibnamefont {Jiang}}, \bibinfo {author} {\bibfnamefont
  {L.}~\bibnamefont {Li}}, \bibinfo {author} {\bibfnamefont {N.-L.}\
  \bibnamefont {Liu}}, \bibinfo {author} {\bibfnamefont {C.-Y.}\ \bibnamefont
  {Lu}}, \bibinfo {author} {\bibfnamefont {A.}~\bibnamefont {Zeilinger}},\ and\
  \bibinfo {author} {\bibfnamefont {J.-W.}\ \bibnamefont {Pan}},\ }\bibfield
  {title} {\bibinfo {title} {Quantum teleportation in high dimensions},\ }\href
  {https://doi.org/10.1103/PhysRevLett.123.070505} {\bibfield  {journal}
  {\bibinfo  {journal} {Phys. Rev. Lett.}\ }\textbf {\bibinfo {volume} {123}},\
  \bibinfo {pages} {070505} (\bibinfo {year} {2019})}\BibitemShut {NoStop}%
\bibitem [{\citenamefont {T{\'{o}}th}\ and\ \citenamefont
  {Apellaniz}(2014)}]{Toth2014QM}%
  \BibitemOpen
  \bibfield  {author} {\bibinfo {author} {\bibfnamefont {G.}~\bibnamefont
  {T{\'{o}}th}}\ and\ \bibinfo {author} {\bibfnamefont {I.}~\bibnamefont
  {Apellaniz}},\ }\bibfield  {title} {\bibinfo {title} {Quantum metrology from
  a quantum information science perspective},\ }\href
  {https://doi.org/10.1088/1751-8113/47/42/424006} {\bibfield  {journal}
  {\bibinfo  {journal} {J. Phys. A: Math. Theor.}\ }\textbf {\bibinfo {volume}
  {47}},\ \bibinfo {pages} {424006} (\bibinfo {year} {2014})}\BibitemShut
  {NoStop}%
\bibitem [{\citenamefont {Pirandola}\ \emph {et~al.}(2017)\citenamefont
  {Pirandola}, \citenamefont {Laurenza}, \citenamefont {Ottaviani},\ and\
  \citenamefont {Banchi}}]{Pirandola2017}%
  \BibitemOpen
  \bibfield  {author} {\bibinfo {author} {\bibfnamefont {S.}~\bibnamefont
  {Pirandola}}, \bibinfo {author} {\bibfnamefont {R.}~\bibnamefont {Laurenza}},
  \bibinfo {author} {\bibfnamefont {C.}~\bibnamefont {Ottaviani}},\ and\
  \bibinfo {author} {\bibfnamefont {L.}~\bibnamefont {Banchi}},\ }\bibfield
  {title} {\bibinfo {title} {Fundamental limits of repeaterless quantum
  communications},\ }\href {https://doi.org/10.1038/ncomms15043} {\bibfield
  {journal} {\bibinfo  {journal} {Nat. Comm.}\ }\textbf {\bibinfo {volume}
  {8}},\ \bibinfo {pages} {15043} (\bibinfo {year} {2017})}\BibitemShut
  {NoStop}%
\bibitem [{\citenamefont {Pereira}\ and\ \citenamefont
  {Pirandola}(2021)}]{Pereira2021Cryp}%
  \BibitemOpen
  \bibfield  {author} {\bibinfo {author} {\bibfnamefont {J.~L.}\ \bibnamefont
  {Pereira}}\ and\ \bibinfo {author} {\bibfnamefont {S.}~\bibnamefont
  {Pirandola}},\ }\bibfield  {title} {\bibinfo {title} {Bounds on
  amplitude-damping-channel discrimination},\ }\href
  {https://doi.org/10.1103/PhysRevA.103.022610} {\bibfield  {journal} {\bibinfo
   {journal} {Phys. Rev. A}\ }\textbf {\bibinfo {volume} {103}},\ \bibinfo
  {pages} {022610} (\bibinfo {year} {2021})}\BibitemShut {NoStop}%
\bibitem [{\citenamefont {Streltsov}\ \emph {et~al.}(2018)\citenamefont
  {Streltsov}, \citenamefont {Kampermann}, \citenamefont {W\"{o}lk}, \citenamefont
  {Gessner},\ and\ \citenamefont {Bru{\ss}}}]{Streltsov2018MaxC}%
  \BibitemOpen
  \bibfield  {author} {\bibinfo {author} {\bibfnamefont {A.}~\bibnamefont
  {Streltsov}}, \bibinfo {author} {\bibfnamefont {H.}~\bibnamefont
  {Kampermann}}, \bibinfo {author} {\bibfnamefont {S.}~\bibnamefont {W\"{o}lk}},
  \bibinfo {author} {\bibfnamefont {M.}~\bibnamefont {Gessner}},\ and\ \bibinfo
  {author} {\bibfnamefont {D.}~\bibnamefont {Bru{\ss}}},\ }\bibfield  {title}
  {\bibinfo {title} {Maximal coherence and the resource theory of purity},\
  }\href {https://doi.org/10.1088/1367-2630/aac484} {\bibfield  {journal}
  {\bibinfo  {journal} {New J. Phys.}\ }\textbf {\bibinfo {volume} {20}},\
  \bibinfo {pages} {053058} (\bibinfo {year} {2018})}\BibitemShut {NoStop}%
\bibitem [{\citenamefont {Chitambar}\ and\ \citenamefont
  {Gour}(2019)}]{Chitambar-QRTs}%
  \BibitemOpen
  \bibfield  {author} {\bibinfo {author} {\bibfnamefont {E.}~\bibnamefont
  {Chitambar}}\ and\ \bibinfo {author} {\bibfnamefont {G.}~\bibnamefont
  {Gour}},\ }\bibfield  {title} {\bibinfo {title} {Quantum resource theories},\
  }\href {https://doi.org/10.1103/RevModPhys.91.025001} {\bibfield  {journal}
  {\bibinfo  {journal} {Rev. Mod. Phys.}\ }\textbf {\bibinfo {volume} {91}},\
  \bibinfo {pages} {025001} (\bibinfo {year} {2019})}\BibitemShut {NoStop}%
\bibitem [{\citenamefont {Horn}\ and\ \citenamefont
  {Johson}(2013)}]{Horn-GramMatrix}%
  \BibitemOpen
  \bibfield  {author} {\bibinfo {author} {\bibfnamefont {R.~A.}\ \bibnamefont
  {Horn}}\ and\ \bibinfo {author} {\bibfnamefont {C.~R.}\ \bibnamefont
  {Johson}},\ }\href@noop {} {\emph {\bibinfo {title} {Matrix Analysis}}},\
  \bibinfo {edition} {2nd}\ ed.\ (\bibinfo  {publisher} {Cambridge University
  Press},\ \bibinfo {year} {2013})\BibitemShut {NoStop}%
\bibitem [{\citenamefont {Nielsen}(1999)}]{Nielsen-Maj}%
  \BibitemOpen
  \bibfield  {author} {\bibinfo {author} {\bibfnamefont {M.~A.}\ \bibnamefont
  {Nielsen}},\ }\bibfield  {title} {\bibinfo {title} {Conditions for a class of
  entanglement transformations},\ }\href
  {https://doi.org/10.1103/PhysRevLett.83.436} {\bibfield  {journal} {\bibinfo
  {journal} {Phys. Rev. Lett.}\ }\textbf {\bibinfo {volume} {83}},\ \bibinfo
  {pages} {436} (\bibinfo {year} {1999})}\BibitemShut {NoStop}%
\bibitem [{\citenamefont {Du}\ \emph {et~al.}(2015)\citenamefont {Du},
  \citenamefont {Bai},\ and\ \citenamefont {Guo}}]{Du-CoherenceMaj}%
  \BibitemOpen
  \bibfield  {author} {\bibinfo {author} {\bibfnamefont {S.}~\bibnamefont
  {Du}}, \bibinfo {author} {\bibfnamefont {Z.}~\bibnamefont {Bai}},\ and\
  \bibinfo {author} {\bibfnamefont {Y.}~\bibnamefont {Guo}},\ }\bibfield
  {title} {\bibinfo {title} {Conditions for coherence transformations under
  incoherent operations},\ }\href {https://doi.org/10.1103/PhysRevA.91.052120}
  {\bibfield  {journal} {\bibinfo  {journal} {Phys. Rev. A}\ }\textbf {\bibinfo
  {volume} {91}},\ \bibinfo {pages} {052120} (\bibinfo {year}
  {2015})}\BibitemShut {NoStop}%
\bibitem [{\citenamefont {Bhatia}(1997)}]{Bhatia}%
  \BibitemOpen
  \bibfield  {author} {\bibinfo {author} {\bibfnamefont {R.}~\bibnamefont
  {Bhatia}},\ }\href@noop {} {\emph {\bibinfo {title} {Matrix Analysis}}}\
  (\bibinfo  {publisher} {Springer-Verlag},\ \bibinfo {year}
  {1997})\BibitemShut {NoStop}%
\bibitem [{\citenamefont {Varga}(2000)}]{VargaMatrixIterative2000}%
  \BibitemOpen
  \bibfield  {author} {\bibinfo {author} {\bibfnamefont {R.~S.}\ \bibnamefont
  {Varga}},\ }\href {https://doi.org/10.1007/978-3-642-05156-2} {\emph
  {\bibinfo {title} {Matrix {{Iterative Analysis}}}}},\ \bibinfo {series}
  {Springer {{Series}} in {{Computational Mathematics}}}, Vol.~\bibinfo
  {volume} {27}\ (\bibinfo  {publisher} {{Springer Berlin Heidelberg}},\
  \bibinfo {address} {{Berlin, Heidelberg}},\ \bibinfo {year}
  {2000})\BibitemShut {NoStop}%
\bibitem [{\citenamefont {Regula}\ \emph {et~al.}(2018)\citenamefont {Regula},
  \citenamefont {Fang}, \citenamefont {Wang},\ and\ \citenamefont
  {Adesso}}]{2018Regula-OneShotCohDis}%
  \BibitemOpen
  \bibfield  {author} {\bibinfo {author} {\bibfnamefont {B.}~\bibnamefont
  {Regula}}, \bibinfo {author} {\bibfnamefont {K.}~\bibnamefont {Fang}},
  \bibinfo {author} {\bibfnamefont {X.}~\bibnamefont {Wang}},\ and\ \bibinfo
  {author} {\bibfnamefont {G.}~\bibnamefont {Adesso}},\ }\bibfield  {title}
  {\bibinfo {title} {One-shot coherence distillation},\ }\href
  {https://doi.org/10.1103/PhysRevLett.121.010401} {\bibfield  {journal}
  {\bibinfo  {journal} {Phys. Rev. Lett.}\ }\textbf {\bibinfo {volume} {121}},\
  \bibinfo {pages} {010401} (\bibinfo {year} {2018})}\BibitemShut {NoStop}%
\bibitem [{\citenamefont {Torun}\ \emph {et~al.}(2019)\citenamefont {Torun},
  \citenamefont {Lami}, \citenamefont {Adesso},\ and\ \citenamefont
  {Yildiz}}]{2019Torunn-CohDis}%
  \BibitemOpen
  \bibfield  {author} {\bibinfo {author} {\bibfnamefont {G.}~\bibnamefont
  {Torun}}, \bibinfo {author} {\bibfnamefont {L.}~\bibnamefont {Lami}},
  \bibinfo {author} {\bibfnamefont {G.}~\bibnamefont {Adesso}},\ and\ \bibinfo
  {author} {\bibfnamefont {A.}~\bibnamefont {Yildiz}},\ }\bibfield  {title}
  {\bibinfo {title} {Optimal distillation of quantum coherence with reduced
  waste of resources},\ }\href {https://doi.org/10.1103/PhysRevA.99.012321}
  {\bibfield  {journal} {\bibinfo  {journal} {Phys. Rev. A}\ }\textbf {\bibinfo
  {volume} {99}},\ \bibinfo {pages} {012321} (\bibinfo {year}
  {2019})}\BibitemShut {NoStop}%
\bibitem [{\citenamefont {Liu}\ and\ \citenamefont
  {Zhou}(2019)}]{2019Liu-DetCoh}%
  \BibitemOpen
  \bibfield  {author} {\bibinfo {author} {\bibfnamefont {C.~L.}\ \bibnamefont
  {Liu}}\ and\ \bibinfo {author} {\bibfnamefont {D.~L.}\ \bibnamefont {Zhou}},\
  }\bibfield  {title} {\bibinfo {title} {Deterministic coherence
  distillation},\ }\href {https://doi.org/10.1103/PhysRevLett.123.070402}
  {\bibfield  {journal} {\bibinfo  {journal} {Phys. Rev. Lett.}\ }\textbf
  {\bibinfo {volume} {123}},\ \bibinfo {pages} {070402} (\bibinfo {year}
  {2019})}\BibitemShut {NoStop}%
\bibitem [{\citenamefont {Regula}\ \emph {et~al.}(2020)\citenamefont {Regula},
  \citenamefont {Bu}, \citenamefont {Takagi},\ and\ \citenamefont
  {Liu}}]{2020RegulaOneShotDis}%
  \BibitemOpen
  \bibfield  {author} {\bibinfo {author} {\bibfnamefont {B.}~\bibnamefont
  {Regula}}, \bibinfo {author} {\bibfnamefont {K.}~\bibnamefont {Bu}}, \bibinfo
  {author} {\bibfnamefont {R.}~\bibnamefont {Takagi}},\ and\ \bibinfo {author}
  {\bibfnamefont {Z.-W.}\ \bibnamefont {Liu}},\ }\bibfield  {title} {\bibinfo
  {title} {Benchmarking one-shot distillation in general quantum resource
  theories},\ }\href {https://doi.org/10.1103/PhysRevA.101.062315} {\bibfield
  {journal} {\bibinfo  {journal} {Phys. Rev. A}\ }\textbf {\bibinfo {volume}
  {101}},\ \bibinfo {pages} {062315} (\bibinfo {year} {2020})}\BibitemShut
  {NoStop}%
\bibitem [{\citenamefont {Pang}\ and\ \citenamefont {Zhao}(2020)}]{Pang2020}%
  \BibitemOpen
  \bibfield  {author} {\bibinfo {author} {\bibfnamefont {Z.-Y.}\ \bibnamefont
  {Pang}}\ and\ \bibinfo {author} {\bibfnamefont {M.-J.}\ \bibnamefont
  {Zhao}},\ }\bibfield  {title} {\bibinfo {title} {Probabilistic coherence
  distillation with assisted setting},\ }\href
  {https://doi.org/10.1007/s11128-020-02857-5} {\bibfield  {journal} {\bibinfo
  {journal} {Quantum Inf. Process.}\ }\textbf {\bibinfo {volume} {19}},\
  \bibinfo {pages} {363} (\bibinfo {year} {2020})}\BibitemShut {NoStop}%
\end{thebibliography}

%

\end{document}